\documentclass[%
 reprint,
 amsmath,amssymb,
 aps,
 nofootinbib,
 superscriptaddress
]{revtex4-2}

\usepackage{amsthm}
\usepackage{graphicx}
\usepackage{xcolor}
\usepackage{braket}
\usepackage{amsmath}
\usepackage{amssymb}
\usepackage{enumerate}
\usepackage{booktabs}
\usepackage{txfonts}
\usepackage[subrefformat=parens]{subcaption}
\usepackage{caption}
\usepackage{mathtools}
\usepackage{etoolbox}
\usepackage[linesnumbered,ruled,vlined]{algorithm2e}
\usepackage{bbm}
\usepackage{here}

\usepackage[
colorlinks=true,
urlcolor=blue,
citecolor=blue,
linkcolor=blue,
hyperfootnotes=false,
breaklinks=true
]{hyperref}

\newtheorem{theorem}{Theorem}
\newtheorem{lemma}{Lemma}
\newtheorem{problem}{Problem}
\newtheorem{assumption}{Assumption}
\newtheorem{example}{Example}

\newcommand{\proc}[1]{\textup{\textsf{#1}}}
\newcommand{\proceduresname}{Procedure}

\newcommand{\algorithmsname}{Algorithm}

\makeatletter
\newcommand*{\rom}[1]{\expandafter\@slowromancap\romannumeral #1@}
\makeatother

\makeatletter
\AtBeginDocument{%
\@ifpackageloaded{hyperref}{%
  \let\orighyper@refstepcounter\hyper@refstepcounter%
  \newcommand{\specialhyper@refstepcounter}[1]{%
    \orighyper@refstepcounter{\@specialcountername}%
    \renewcommand{\@currentHref}{\@specialcountername.\csname the\@specialcountername\endcsname}%
  }
}{%
  \providecommand{\autoref}[1]{\ref{#1}}
  \newcommand{\specialhyper@refstepcounter}[1]{%
  }
}

\newcommand{\specialcountername}[1]{%
  \def\@specialcountername{#1}
}

\AtBeginEnvironment{algorithm}{%
\specialcountername{procedures}%
\edef\temp@@a{\csname\@specialcountername name\endcsname}
\edef\temp@@b{\algorithmcfname}
\ifx\temp@@a\temp@@b
\csletcs{c@algocf}{c@\@specialcountername}
\csletcs{thealgocf}{the\@specialcountername}
\let\hyper@refstepcounter\specialhyper@refstepcounter
\fi
}
}
\makeatother

\begin{document}

\title{Quantum algorithm for position weight matrix matching}
\author{Koichi Miyamoto}
\email{miyamoto.kouichi.qiqb@osaka-u.ac.jp}
\affiliation{Center for Quantum Information and Quantum Biology, Osaka University, Toyonaka, Osaka 560-0043, Japan}

\author{Naoki Yamamoto}
\affiliation{Department of Applied Physics and Physico-Informatics, Keio University, Yokohama, Kanagawa, Japan}
\affiliation{Quantum Computing Center, Keio University, Yokohama, Kanagawa, Japan}

\author{Yasubumi Sakakibara}
\affiliation{Department of Biosciences and Informatics, Keio University, Yokohama, Kanagawa, Japan}
\affiliation{Quantum Computing Center, Keio University, Yokohama, Kanagawa, Japan}

\date{\today}

\begin{abstract}

We propose two quantum algorithms for a problem in bioinformatics, position weight matrix (PWM) matching, which aims to find segments (sequence motifs) in a biological sequence such as DNA and protein that have high scores defined by the PWM and are thus of informational importance related to biological function.
The two proposed algorithms, the naive iteration method and the 
Monte-Carlo-based method, output matched segments, given the oracular accesses to the entries in the biological sequence and the PWM.
The former uses quantum amplitude amplification (QAA) for sequence motif search, resulting in the query complexity scaling on the sequence length $n$, the sequence motif length $m$ and the number of the PWMs $K$ as $\widetilde{O}\left(m\sqrt{Kn}\right)$, which means speedup over existing classical algorithms with respect to $n$ and $K$.
The latter also uses QAA, and further, quantum Monte Carlo integration for segment score calculation, instead of iteratively operating quantum circuits for arithmetic in the naive iteration method; then it provides the additional speedup with respect to $m$ in some situation.
As a drawback, these algorithms use quantum random access memories and their initialization takes $O(n)$ time.
Nevertheless, our algorithms keep the advantage especially when we search matches in a sequence for many PWMs in parallel.

\end{abstract}

\maketitle

\section{Introduction}

Quantum computing \cite{nielsen_chuang_2010} is an emerging technology that has a potential to provide large benefits to various fields.
Many quantum algorithms that speedup time-consuming problems in classical computing have been proposed, and their applications to practical problems in industry and science have been studied.
In this paper, we focus on an important problem in bioinformatics: {\it position weight matrix matching}.

As a field in bioinformatics, {\it sequence analysis}, which focuses on biological sequences such as DNA sequences and protein amino-acid sequences, has a long history.
Such a sequence is represented as a string of {\it alphabets} $\mathcal{A}$, e.g. $\{A,G,T,C\}$ for nucleabases in DNA and 20 letters for 20 types of amino acids in a protein, and holds biological information.
As a tool to extract important information from a sequence, position weight matrices (PWMs), also known as position specific scoring matrices (PSSMs), are often used.
More concretely, a PWM is a tool to find segments with fixed length $m$ that seems to hold specific information from a sequence with length $n$.
As explained in details later, a PWM $M$ is a matrix of real values, and its entries determine a score given to a segment.
For example, if the $i$th position in a segment has the $j$th alphabet, the $(i,j)$th entry of $M$ is the score of that position.
The score of the segment is defined as the sum of the scores at all the positions.
We then search segments that have scores higher than the predetermined threshold.
In this way, we can find some specific patterns ({\it sequence motifs}) in a sequence, admitting not only exact match but also approximate match to some extent.
PWMs is in fact used to, for example, find transcription factor binding sites in DNA \cite{stormo1982use} and infer the three-dimensional structure of a protein \cite{doi:10.1073/pnas.84.13.4355}.

Following recent developments of next-generation DNA sequencing technology, the volume of data handled in sequence analysis is exponentially growing.
Although many classical algorithms and tools for PWM matching have been devised so far \cite{quandt1995matlnd,10.1093/bioinformatics/16.3.233,10.5555/645635.660986,rajasekaran2002efficient,10.1093/nar/gkg585,mci/Beckstette2004,freschi2005using,beckstette2006fast,10.1007/11780441_36,10.1007/978-3-540-71233-6_19,10.1093/bioinformatics/btp554,4803829,fostier2020blamm}, it is interesting to investigate the potential of novel technologies such as quantum computing to speedup this numerical problem to the extent that classical algorithms cannot reach.

Based on such a motivation, in this paper, we propose two quantum algorithms for PWM matching.
As far as the authors know, this is the first proposal on quantum algorithms for this problem, although there are some quantum algorithms for exact string match \cite{RAMESH2003103,montanaro2017quantum,niroula2021quantum}.

Our quantum algorithms are two-fold: calculating scores of segments and searching high-score segments.
For the latter, we use quantum amplitude amplification (QAA) \cite{brassard2002}, which is an extension of Grover's algorithm for unstructured database search \cite{Grover1996}.
As is well known, this provides the quadratic quantum speedup with respect to the number of searched data points, which now corresponds to $n$ the length of the sequence.

On the former, we consider two approaches for score calculation, which differentiate the two proposed algorithms.
The first one is calculating the segment score by adding the position-wise scores one by one using the quantum circuits for arithmetic.
We name the PWM matching quantum algorithm based on this approach the {\it naive iteration method}.
By this method, for any sequence with length $n$ and any $K$ PWMs for sequence motifs with length $m$, given the oracles to get the specified entry in them, we can find $n_{\rm sol}$ matches with high probability making $\widetilde{O}\left(m\sqrt{Knn_{\rm sol}}\right)$ queries\footnote{The symbol $\widetilde{O}(\cdot)$ hides logarithmic factors in $O(\cdot)$.} to the oracles.
As far as the authors know, there is no known classical PWM matching algorithm whose worst-case complexity is sublinear to $n$, and thus the above complexity just shows the quantum speedup.
Moreover, note that we aim to search the matches for the multiple PWMs at the same time, which has not been considered in classical algorithms.
We achieve the quantum speedup with respect to $K$ too, compared with the $K$-times sequential runs of the algorithm for the different PWMs, whose complexity obviously scales with $K$ linearly. 
 
The second quantum algorithm uses the quantum Monte Carlo integration (QMCI) \cite{Montanaro2015} to calculate the segment score; we call this method the {\it QMCI-based method}. 
QMCI is, similarly to classical Monte Carlo, the method to estimate expectations of random variables, integrals and sums, and also provides the quadratic speedup compared with the classical counterpart.
We therefore use this to calculate the segment score, which is the sum of the position-wise scores, expecting the further speedup from the naive iteration method, especially when $m$ is large and thus the sum has many terms.
This combination of QMCI and QAA is similar to the quantum algorithm for gravitational wave (GW) matched filtering proposed in \cite{Miyamoto2022}.
A drawback of this approach is the possibility of false detection: the result of QMCI inevitably has an error and the erroneous estimate on the segment score can exceed the threshold even if the true score does not.
To cope with this issue, we introduce two levels of the threshold, $w_{\rm soft}$ and $w_{\rm hard}$, which have the following meaning: we never want to miss segments with scores higher than $w_{\rm hard}$ or falsely find those with scores lower than $w_{\rm soft}$, and it is not necessary but good to find those with scores between $w_{\rm soft}$ and $w_{\rm hard}$.
In this setting, the QMCI-based method outputs all the segments with scores higher than $w_{\rm hard}$ possibly along with some of those with in-between scores with high probability, making $\widetilde{O}\left(\frac{mn_{\rm soft}\sqrt{Kn}}{w_{\rm hard}-w_{\rm soft}}\right)$ oracle calls at most, where $n_{\rm soft}$ is the number of the segments with scores higher than $w_{\rm soft}$.
Although this complexity seemingly has the same dependency on $m$ as the naive iteration method, it can be sublinear to $m$, since, as explained later, a reasonable choice of $w_{\rm soft}$ and $w_{\rm hard}$ is such that $w_{\rm hard}-w_{\rm soft}\sim\sqrt{m}$.

Although the complexities of the proposed quantum algorithms are sublinear to $n$ and/or $m$, they require time for preparation.
The oracles used in the algorithms can be implemented by quantum random access memory (QRAM) \cite{Giovannetti2008QRAM}, and initialization of QRAM, that is, registering the values of the entries in the sequence and the PWMs takes time, which is estimated as $O(n)$ in usual situations.
Despite of this initialization cost, our quantum algorithms still have the advantage over the existing classical algorithms, since, among classical ones with $O(n)$ initialization cost, none has the worst-case complexity sublinear to $n$ in the main part.
Also note that, once we prepare the QRAM for a sequence, our quantum algorithms can search the matches between that sequence and many PWMs, with much smaller initialization cost of the QRAM for the PWMs.

The rest of this paper is organized as follows.
Sec. \ref{sec:pre} is preliminary one, where we introduce PWM matching and some building-block quantum algorithms such as QAA and QMCI.
Sec. \ref{sec:QPWMM} is the main part, where we explain our quantum algorithms for PWM matching, the naive iteration method and the QMCI-based method, presenting the detailed procedures in them and the estimations on their complexities.
In Sec. \ref{sec:dis}, we discuss the aforementioned issues on our algorithms, the preparation cost for the QRAMs and the plausible setting
on the segment score thresholds.
Sec. \ref{sec:sum} summarizes this paper.

\section{Preliminary \label{sec:pre}}

\subsection{Notation}

We denote the set of all positive real numbers by $\mathbb{R}_+$ and the set of all nonnegative real numbers by $\mathbb{R}_{\ge 0}$.

For each $n\in\mathbb{N}$, we define $[n]:=\{1,...,n\}$, $[n]_0:=\{0,...,n-1\}$ and $\mathbb{N}_{\ge n}:=\{m\in\mathbb{N} \ | \ m\ge n\}$.

For any probability space $(\Omega,\mathcal{F},\mathbb{P})$ and any random variable $X$ on it, we denote the expectation of $X$ by $E_\mathbb{P}[X]$.

For any finite set $\mathcal{S}$ and any $n\in\mathbb{N}$, we denote $S=(s_0,...,s_{n-1})\in\mathcal{S}^n$, where $s_i\in\mathcal{S}$ for each $i\in[n]_0$, by $S=s_0..s_{n-1}$.

For any equation or inequality $C$, $\mathbbm{1}_C$ takes 1 if $C$ is satisfied, and 0 otherwise.

For any $x\in\mathbb{R}$, if $|x-y|\le\epsilon$ holds for some $y\in\mathbb{R}$ and $\epsilon\in\mathbb{R}_+$, we say that $x$ is an $\epsilon$-approximation of $y$.

\subsection{Position weight matrix matching \label{sec:PWM}}


Here, we formally define the problem we hereafter consider.

\begin{problem}[PWM matching]
    Suppose that we are given the following:
    \begin{itemize}
        \item A finite set $\mathcal{A}$ called the alphabet, whose elements are labeled with integers in $[|\mathcal{A}|]_0$.
        
        \item $K$ matrices $M_k=(M_k(i,a))_{i\in[m]_0,a\in\mathcal{A}}\in\mathbb{R}^{m\times|\mathcal{A}|},k\in[K]_0$ called the PWMs, where $K\in\mathbb{N}$ and $m\in\mathbb{N}_{\ge 2}$ is called the {\it length} of the PWM.
        
        \item An element $S=s_0..s_{n-1}$ in $\mathcal{A}^n$, where $n$ is an integer larger than $m$.
        We call it a sequence.
        \item $w_{\rm th}\in\mathbb{R}$ called the threshold.
    \end{itemize}
    For each $(k,i)\in\mathcal{P}_{\rm all}:=[K]_0\times[n-m+1]_0$, define
    \begin{equation}
        w_{k,i}:=W_{M_k}(s_i..s_{i+m-1}),
    \end{equation}
    where, for every $M=(M(i,a))_{i\in[m]_0,a\in\mathcal{A}}\in\mathbb{R}^{m\times|\mathcal{A}|}$ and $u_0..u_{m-1}\in\mathcal{A}^m$,
    \begin{equation}
        W_M(u_0..u_{m-1}):=\sum_{j=0}^{m-1} M(j,u_j).
    \end{equation}
    Then, we want to find all the elements in the set
    \begin{equation}
        \mathcal{P}_{\rm sol}:=\{(k,i)\in\mathcal{P}_{\rm all} \ | \ w_{k,i}\ge w_{\rm th}\}.
        \label{eq:thres}
    \end{equation}
    \label{prob:PWMM}
\end{problem}

\begin{example}[PWM score calculation]
  {\rm 
  An example of calculating scores for segments using PWM is shown below.
  The following PWM of length 8 represents the binding site motif for a transcription factor: \\
  
    \begin{tabular}{c|cccccccc}
      \texttt{A} & $-1.31$ & $-0.62$ & $-1.31$ & $+0.63$ & $-1.31$ & $-1.31$ & $-1.31$ & $+0.48$ \\
      \texttt{C} & $-0.83$ & $-0.83$ & $+1.12$ & $-0.83$ & $+1.12$ & $+1.12$ & $+1.37$ & $-0.83$ \\
      \texttt{G} & $-0.83$ & $+1.25$ & $+0.27$ & $+0.27$ & $-0.83$ & $+0.27$ & $-0.83$ & $+0.56$ \\
      \texttt{T} & $+0.89$ & $-1.31$ & $-1.31$ & $-1.31$ & $-0.21$ & $-1.31$ & $-1.31$ & $-1.31$ \\
      \hline 
      & (1) & (2) & (3) & (4) & (5) & (6) & (7) & (8)  
    \end{tabular}
    \\
    
    Given this PWM $M$, the score $W_M$ for the DNA sequence (segment) ``TACATGCA'' is calculated as follows:
\begin{eqnarray*}
  & & W_M(\mathtt{TACATGCA}) \\ 
  & & = \ \overbrace{+0.89}^{\mathtt{T}}
  \overbrace{-0.62}^{\mathtt{A}}
  \overbrace{+1.12}^{\mathtt{C}}
  \overbrace{+0.63}^{\mathtt{A}}
  \overbrace{-0.21}^{\mathtt{T}}
  \overbrace{+0.27}^{\mathtt{G}}
  \overbrace{+1.37}^{\mathtt{C}}
  \overbrace{+0.48}^{\mathtt{A}} \\
  & & = \ 3.93 
\end{eqnarray*}
}

{\rm This PWM matching is then applied to a long genome DNA sequence of million bases such that every segment $i$ in the DNA sequence is assigned a score $W_M(u_i\ldots u_{i+m-1})$ and we search $\mathcal{P}_{\rm sol}$, segments with scores higher than the threshold $w_{\rm th}$. }
    \label{example:PWM}
\end{example}

In the Problem \ref{prob:PWMM}, we consider the match with the multiple PWMs simultaneously ($K\ge 2$), although Example \ref{example:PWM} is a single-PWM case ($K=1$).
In general, the DNA sequence of a genome contains hundreds of sequence motifs and the annotation for a genome sequence must be completed by finding all sequence motifs using multiple PWMs simultaneously.
As we will see later, we can achieve the quantum speedup with respect to the number $K$ of mulitple PWMs; that is, our quantum algorithm finds all the matches between the sequence $S$ and the multiple PWMs $\{M_k\}_{k=0}^{K-1}$ faster than iterating individually searching for the matches between $S$ and each PWM.
One might concern that, although it is assumed that all the $K$ PWSs have same length, this is not always the case.
This point is easily settled as follows.
Denoting the lengths of $M_0,...,M_{K-1}$ by $m_0,...,m_{K-1}$, respectively, we set $m:=\max\{m_0,...,m_{K-1}\}$ and, for each $k\in[K]_0$, replace $M_k$ with the $m\times |\mathcal{A}|$ matrix whose first to $m_k$th rows are those in the original $M_k$ and $(m_k+1)$th to $m$th rows are filled with 0.
Note that the score of each segment in $S$ does not change under this modification\footnote{Note that, for $k\in[K]_0$ such that $m_k<m$, this modification makes the last $m-m_k$ segments with length $m_k$ out of the scope of the matching, although they should be considered. We calculate the scores for such segments and check whether they exceed the threshold, separately from our algorithm. We can reasonably assume that these additional calculations and checks take the negligible time, as far as the number of these exceptional segments, $m-m_k$, is much smaller than that of all the segment, $n-m_k+1$.}.

The typical orders of magnitudes of the parameters in PWM matching are as follows.
The sequence length $n$ can be of order $10^8$ (resp. $10^6-10^7$) and the number of PWMs $K$ may be of order $10^2$ (resp. $10^4$) for DNA (resp. protein) \cite{beckstette2006fast}.
The sequence motif length $m$ is typically about ten or several tens \cite{beckstette2006fast}, but motifs with lengths greater than $10^2$ are sometimes considered for protein \cite{shameer20093pfdb}.

Hereafter, we assume that entries of PWMs are bounded:
\begin{equation}
0\le M_k(i,a) \le 1 \label{eq:M01}
\end{equation}
for every $k\in[K]_0$, $i\in[m]_0$ and $a\in\mathcal{A}$.
This is just for the later convenience in using QMCI for score calculation.
Although this condition is not satisfied in general cases including Example \ref{example:PWM}, we can meet it by rescaling. 
That is, we redefine $M^\prime_k$ as $M_k$, where
\begin{equation}
    M_k^\prime(i,a) := \frac{M_k(i,a)-M_{\rm min}}{M_{\rm max}-M_{\rm min}}
\end{equation}
with
\begin{eqnarray}
    M_{\rm max} &:=& \max_{(k,i,a)\in[K]_0\times[m]_0\times\mathcal{A}} M_k(i,a), \nonumber \\
    M_{\rm min} &:=& \min_{(k,i,a)\in[K]_0\times[m]_0\times\mathcal{A}} M_k(i,a).
\end{eqnarray}
It is easy to see that after this redefinition Eq. (\ref{eq:M01}) holds.
We also need to replace the threshold $w_{\rm th}$ with
\begin{equation}
    w^\prime_{\rm th}:=\frac{w_{\rm th}-mM_{\rm min}}{M_{\rm max}-M_{\rm min}}.
\end{equation}
Note that the set (\ref{eq:thres}) is invariant under the above rescaling.

\subsection{Quantum algorithms}

Here, we briefly explain the building-block quantum algorithms for our PWM matching algorithm.

\subsubsection{Arithmetic on a quantum computer}

Before introducing the quantum algorithms, let us summarize the setup for quantum computation and the elementary operations we use in this paper.

We consider computation on the system consisting of the multiple quantum registers. 
We treat real numbers in fixed-point binary representation and, for each $x\in\mathbb{R}$, we denote by $\ket{x}$ the computational basis state on a quantum register where the bit string on the register corresponds to the binary representation of $x$.
We assume that each register has a sufficient number of qubits and thus the error from finite-precision representation is negligible.

We use the oracles for elementary arithmetic such as the adder $O_{\rm add}\ket{x}\ket{y}=\ket{x}\ket{x+y}$, the subtracter $O_{\rm sub}\ket{x}\ket{y}=\ket{x-y}\ket{y}$ and the multiplier $O_{\rm mul}\ket{x}\ket{y}=\ket{x}\ket{xy}$, where $x,y\in\mathbb{R}$.
Many proposals on implementations of such oracles have been made so far: see \cite{Munoz2022} and the references therein.

Besides, we also assume the availability of the following oracles.
The oracle $O_=$ checks whether two numbers are equal or not: for any $x,y\in\mathbb{R}$, $O_{=}\ket{x}\ket{y}\ket{0}=\ket{x}\ket{y}\left(\mathbbm{1}_{x=y}\ket{0}+\mathbbm{1}_{x\ne y}\ket{1}\right)$.
Also, the comparator $O_{\rm comp}$ acts as $O_{\rm comp}\ket{x}\ket{y}\ket{0}=\ket{x}\ket{y}\left(\mathbbm{1}_{x\ge y}\ket{1}+\mathbbm{1}_{x<y}\ket{0}\right)$ for any $x,y\in\mathbb{R}$.
These oracles can be implemented via subtraction.
To check $x=y$ or not, we may calculate $x-y$ and see whether it is 0 or not.
Therefore, we can implement $O_{=}$ by using a subtracter, followed by a multiple controlled-NOT (CNOT) gate activated if and only if all the bits of $x-y$ are 0, and at last a NOT gate.
Moreover, we can implement $O_{\rm comp}$ by combining a substacter with a CNOT gate activated if and only if the most significant bit of $x-y$ is 0; this is because, if we adopt 2’s complement method to represent negative numbers, the most significant bit represents the sign of a number \cite{koren2001}.

In addition, for any $N\in\mathbb{N}_{\rm 2}$, we assume the availability of the oracle $O^{\rm EqPr}_N$ that generates the equiprobable superposition of $\ket{0},\ket{1},...,\ket{N-1}$: $O^{\rm EqPr}_N\ket{0}=\frac{1}{\sqrt{N}}\sum_{i=0}^{N-1} \ket{i}$.
If $N=2^n$ with some $n\in N$, we can implement this oracle just by operating a Hadamard gate on each qubit of the $n$-qubit register.
If not, letting $n$ be $\lceil \log_2 N \rceil$, we can implement $O^{\rm EqPr}_N$ by the method in \cite{grover2002creating} to generate a state in which a given probability density function $p(x)$ is amplitude-encoded, with $p(x)$ defined on $[0,1]$ as $p(x)=\mathbbm{1}_{x\le N/2^n}$.

Lastly, we use the oracle $O^{\rm med}_N$ that outputs the median ${\rm med}(x_1,...,x_N)$ of any $N$ real numbers $x_1,...,x_N$, that is, $O^{\rm med}_N\ket{x_1}\cdots\ket{x_N}\ket{0}=\ket{x_1}\cdots\ket{x_N}\ket{{\rm med}(x_1,...,x_N)}$.
The implementations of this oracle have been discussed in \cite{Miyamoto2022}.

Hereafter, we collectively call the above oracles the arithmetic oracles.

\subsubsection{Quantum amplitude amplification}

The first building block for our PWM matching algorithm is QAA \cite{brassard2002}. 
Given the oracle to generate the superposition state $\ket{\Phi}$, QAA amplifies the amplitude of the ``marked state" in $\ket{\Phi}$ so that we can obtain it quadratically faster than naively iterating the process of generating $\ket{\Phi}$ and 
measurement on it. 
Here, we give the following theorem, which was presented in \cite{Miyamoto2022} 
as a slight modification of the original one in \cite{brassard2002}.

\begin{theorem}[Theorem 2 in \cite{Miyamoto2022}, originally Theorem 3 in \cite{brassard2002}]

    Suppose that we are given an access to an oracle $V$ that acts on the system $R$ consisting of a quantum register $R_1$ and a qubit $R_2$, as
    \begin{equation}
    V\ket{0}\ket{\bar{0}}=\sqrt{a}\ket{\phi_1}\ket{\bar{1}}+\sqrt{1-a}\ket{\phi_0}\ket{\bar{0}}=:\ket{\Phi}, \label{eq:V}
    \end{equation}
    where $\ket{\phi_1}$ and $\ket{\phi_0}$ are some quantum states on $R_1$ and $a\in[0,1)$.
    Then, for any $\gamma,\delta\in(0,1)$, there exists a quantum algorithm $\proc{QAA}(V,\gamma,\delta)$ that behaves as follows:
    \begin{itemize}
        \item The output of the algorithm is either of
        \begin{enumerate}
            \renewcommand{\labelenumi}{(\Alph{enumi})}
            \item the message ``success" and the quantum state $\ket{\phi_1}$
            \item the message ``failure"
        \end{enumerate}
        
        \item If $a\ge\gamma$, the algorithm outputs (A) with probability at least $1-\delta$, making $O\left(\frac{\log\delta^{-1}}{\sqrt{a}}\right)$ queries to $V$.
        
        \item If $0<a<\gamma$, the algorithm outputs either (A) or (B), making $O\left(\frac{\log\delta^{-1}}{\sqrt{\gamma}}\right)$ queries to $V$.

        \item If $a=0$, the algorithm certainly outputs (B), making $O\left(\frac{\log\delta^{-1}}{\sqrt{\gamma}}\right)$ queries to $V$\footnote{Although Theorem 2 in \cite{Miyamoto2022} does not mention the case that $a=0$, the statement on this case is proved in Appendix \ref{sec:casea0}}.
    \end{itemize}
    \label{th:QAA}
\end{theorem}

For the detailed procedure of $\proc{QAA}(V,\gamma,\delta)$ and the proof of Theorem \ref{th:QAA}, see \cite{Miyamoto2022} and the original paper \cite{brassard2002}.

\subsubsection{Quantum Monte Carlo integration}

QAA leads to quantum amplitude estimation (QAE) algorithm \cite{brassard2002}, 
which estimates the amplitude of the marked state; QAE is further extended to 
the quantum algorithm for estimating the expectation of a random variable \cite{Montanaro2015}, which we call QMCI in this paper.
This is the second building block.
There are various versions of QMCI for different situations.
For the PWM matching problem, we can use the following one, which assumes that the variable is bounded.

\begin{theorem}[Theorem 2.3 in \cite{Montanaro2015}, modified]
    Let $N\in\mathbb{N}$ and $\mathcal{X}$ be a set of $N$ real numbers $X_0,...,X_{N-1}$, each of which satisfies $0\le X_i\le1$.
    Suppose that we are given an oracle $O_X$ that acts as
    \begin{equation}
    O_{\mathcal{X}}\ket{i}\ket{0}=\ket{i}\ket{X_i}, \label{eq:OX}
    \end{equation}
    for any $i\in[N]_0$.
    Then, for any $\epsilon$ and $\delta$ in $(0,1)$, there is an oracle $O_{\mathcal{X},\epsilon,\delta}^{\rm mean}$ such that
    \begin{equation}
        O_{\mathcal{X},\epsilon,\delta}^{\rm mean}\ket{0}=\sum_{y\in \mathcal{Y}} \alpha_y\ket{y}, \label{eq:OmeanX}
    \end{equation}
    where some ancillary qubits are undisplayed.
    Here, $\mathcal{Y}$ is a finite set of real numbers that includes a subset $\tilde{\mathcal{Y}}$ consisting of $\epsilon$-approximations of the mean of $X_0,...,X_{N-1}$, 
    \begin{equation}
        \mu:=\frac{1}{N}\sum_{i=0}^{N-1}X_i, \label{eq:mu}
    \end{equation}
    and $\{\alpha_y\}_{y\in\mathcal{Y}}$ are complex numbers satisfying
    \begin{equation}
    \sum_{\tilde{y}\in\tilde{\mathcal{Y}}}|\alpha_{\tilde{y}}|^2\ge 1-\delta.
    \end{equation}
    In $O_{\mathcal{X},\epsilon,\delta}^{\rm mean}$,
    \begin{equation}
        O\left(\frac{1}{\epsilon}\log\left(\frac{1}{\delta}\right)\right) \label{eq:compQMCI}
    \end{equation}
    queries to $O_{\mathcal{X}}$ are made.
    \label{th:QMCI}
\end{theorem}

The proof is presented in Appendix \ref{sec:PrThQMCI}, where the detailed way to construct the oracle $O_{\mathcal{X},\epsilon,\delta}^{\rm mean}$ is also shown.
Note that this theorem is slightly modified from the original one, Theorem 2.3 in \cite{Montanaro2015}, in some points.
First, our Theorem \ref{th:QMCI} is on the algorithm to calculate the average $\mu$ of a sequence and the sum, which is instantly obtained by multiplying the sequence size $N$ to the average. 
On the other hand, Theorem 2.3 in \cite{Montanaro2015} presents the algorithm to calculate the expectation of a random variable, and so calculation of the average and the sum is a special case.
This is sufficient for us, since we use QMCI to calculate the score of a segment in the sequence $S$, which is in fact the sum of the scores of the entries in the segment.
Second, although the algorithm in Theorem 2.3 in \cite{Montanaro2015} outputs the approximation of $\mu$, the above Theorem \ref{th:QMCI} mentions only generating the state in Eq. (\ref{eq:OmeanX}).
Although we obtain the approximation of $\mu$ by measuring the state, we do not do so in our PWM matching algorithm.
This is because our algorithm uses QMCI as a subroutine to calculate the score of each segment in the sequence $S$ in the high-score segment search by QAA.
This modification is similar to that in \cite{Miyamoto2022}, which also presents QMCI with no measurement.
However, QMCI in this paper is different from that in \cite{Miyamoto2022} too, since the former assumes that each $X_i$ is bounded but the latter assumes that the upper bound on the variance of the sequence is given.

\section{Quantum algorithm for PWM matching \label{sec:QPWMM}}

We now present the quantum algorithm for PWM matching.
The basic strategy is as follows: we calculate $\{w_{k,i}\}_{k,i}$ for all the pairs $(k,i)\in\mathcal{P}_{\rm all}$ parallelly in a quantum superposition and find the pairs with high scores by QAA.
We present the two versions of the quantum algorithm, the {\it naive iteration method} and the {\it QMCI-based method}, whose difference is how to calculate $w_{k,i}$.

\subsection{Assumption on the quantum accesses to the sequence and the PWMs}

Before presenting the quantum algorithm, we need to make some assumptions on the available oracles.
First, for score calculation on a quantum computer, we need to load the entries in the sequence $S$ and the PWMs $\{M_k\}_k$ onto 
quantum registers.
This is formally stated as follows.

\begin{assumption}
We have accesses to the following oracles:
\begin{itemize}
    \item $O_{\rm seq}$: for any $i\in[n]_0$,
    \begin{equation}
        O_{\rm seq}\ket{i}\ket{0}=\ket{i}\ket{s_i}. \label{eq:Oseq}
    \end{equation}

    \item $O_{\rm PWM}$: for any $(k,i,a)\in[K]_0\times[m]_0\times\mathcal{A}$,
    \begin{equation}
        O_{\rm PWM}\ket{k}\ket{i}\ket{a}\ket{0}=\ket{k}\ket{i}\ket{a}\ket{M_k(i,a)}. \label{eq:OPWM}
    \end{equation}
\end{itemize}
\label{ass:oracle}
\end{assumption}

Here and hereafter, $\ket{a}$ with $a\in\mathcal{A}$ is regarded as the computational basis state corresponding to the integer that labels $a$.
We can implement these oracles if QRAM \cite{Giovannetti2008QRAM} is available, but non-negligible preprocessing cost is needed.
We will discuss this point in Sec. \ref{sec:QRAM}.

Besides, we assume that we can use the oracle that determines whether a given index pair $(k,i)\in\mathcal{P}_{\rm all}$ is in a given subset 
$\mathcal{P}\subset\mathcal{P}_{\rm all}$ or not.

\begin{assumption}
    For any 
    subset $\mathcal{P}\subset\mathcal{P}_{\rm all}$, we have an access to the oracle $O_{\mathcal{P}}$ that acts as
    \begin{equation}
        O_{\mathcal{P}}\ket{k}\ket{i}\ket{0}=\ket{k}\ket{i}\left(\mathbbm{1}_{(k,i)\in\mathcal{P}}\ket{0}+\mathbbm{1}_{(k,i)\notin\mathcal{P}}\ket{1}\right)
    \end{equation}
    for any $(k,i)\in\mathcal{P}_{\rm all}$. \label{ass:OP}
\end{assumption}

We can also implement this oracle using QRAM as discussed in Sec. \ref{sec:QRAM}.

Since $O_{\rm seq}$, $O_{\rm PWM}$ and $O_{\mathcal{P}}$ are supposed to be realized by QRAM, we hereafter consider the number of queries to them as a metric of the complexity of our algorithm.

\subsection{Algorithm \rom{1}: the naive iteration method}

We then explain the first algorithm, the naive iteration method.

We can calculate $w_{k,i}$ by naively iterating the queries to $O_{\rm seq}$ and $O_{\rm PWM}$ and additions as Procedure \ref{proc:naive}.

\ \\

\SetAlgorithmName{\proceduresname}{}{}
\begin{algorithm}[H]
    \KwIn{$k\in[K]_0,i\in[m]_0$}
    
    Prepare the quantum registers $R_1,...,R_7$ and initialize $R_1$, $R_2$, $R_4$ and the others to $\ket{k}$, $\ket{i}$, $\ket{i}$ and $\ket{0}$, respectively.\\

    \For{$j=0,...,m-1$}{

    Set $R_5$ to $\ket{s_l}$ by $O_{\rm seq}$ with the value on $R_4$ being $l$.\\

    Set $R_6$ to $\ket{M_k(l,a)}$ by $O_{\rm PWM}$ with the values on $R_1$, $R_3$ and $R_5$ being $k$, $l$ and $a$, respectively.\\

    Using $O_{\rm add}$, add the value on $R_6$ to the value on $R_7$.\\

    \If{$j<m-1$}{

    Reset $R_5$ and $R_6$ to $\ket{0}$ using the inverses of $O_{\rm seq}$ and $O_{\rm PWM}$, respectively.\\

    Increment the values on $R_3$ and $R_4$ by 1, using $O_{\rm add}$.\\
    }
    }
    
    \caption{calculate $w_{k,i}$ by naive iteration}\label{proc:naive}
\end{algorithm}

\ \\

In this procedure, the quantum state is transformed as follows:
\begin{eqnarray}    
    &&\ket{k}\ket{i}\ket{0}\ket{i}\ket{0}\ket{0}\ket{0} \nonumber \\
    &\xrightarrow{3}& \ket{k}\ket{i}\ket{0}\ket{i}\ket{s_i}\ket{0}\ket{0} \nonumber \\
    &\xrightarrow{4}& \ket{k}\ket{i}\ket{0}\ket{i}\ket{s_i}\ket{M_k(0,s_i)}\ket{0} \nonumber \\
    &\xrightarrow{5}& \ket{k}\ket{i}\ket{0}\ket{i}\ket{s_i}\ket{M_k(0,s_i)}\ket{M_k(0,s_i)} \nonumber \\
    &\xrightarrow{7}& \ket{k}\ket{i}\ket{0}\ket{i}\ket{0}\ket{0}\ket{M_k(0,s_i)} \nonumber \\
    &\xrightarrow{8}& \ket{k}\ket{i}\ket{1}\ket{i+1}\ket{0}\ket{0}\ket{M_k(0,s_i)} \nonumber \\
    &\xrightarrow{3}& \ket{k}\ket{i}\ket{1}\ket{i+1}\ket{s_{i+1}}\ket{0}\ket{M_k(0,s_i)} \nonumber \\
    &\xrightarrow{4}& \ket{k}\ket{i}\ket{1}\ket{i+1}\ket{s_{i+1}}\ket{M_k(1,s_{i+1})}\ket{M_k(0,s_i)} \nonumber \\
    &\xrightarrow{5}& \ket{k}\ket{i}\ket{1}\ket{i+1}\ket{s_{i+1}}\ket{M_k(1,s_{i+1})}\Ket{\sum_{j=0}^1M_k(j,s_{i+j})} \nonumber \\
    &\xrightarrow{7}& \ket{k}\ket{i}\ket{1}\ket{i+1}\ket{0}\ket{0}\Ket{\sum_{j=0}^1M_k(j,s_{i+j})} \nonumber \\
    &\rightarrow& ... \nonumber \\
    &\xrightarrow{5}& \ket{k}\ket{i}\ket{m-1}\ket{i+m-1}\ket{s_{i+m-1}}\ket{M_k(m-1,s_{i+m-1})} \nonumber \\
    && \quad \otimes \ \Bigg|\underbrace{\sum_{j=0}^{m-1}M_k(j,s_{i+j})}_{=w_{k,i}}\Bigg\rangle. \label{eq:Oscit}
\end{eqnarray}
Here, the numbers on the arrows correspond to the steps in Procedure \ref{proc:naive}.
We denote by $O_{\rm sc,it}$ the quantum circuit for the above operation.

Then, using $O_{\rm sc,it}$, we can construct the quantum algorithm to find the high-score segments.

\begin{theorem}
Consider Problem \ref{prob:PWMM} under Assumptions \ref{ass:oracle} and \ref{ass:OP}.
Suppose that we are given $\delta\in(0,1)$.
Then, there exists a quantum algorithm that behaves as follows.
\begin{enumerate}
    \renewcommand{\labelenumi}{(\roman{enumi})}
    \item If $n_{\rm sol}:=|\mathcal{P}_{\rm sol}|>0$, the algorithm outputs all the elements in $\mathcal{P}_{\rm sol}$ with probability at least $1-\delta$, making
    \begin{equation}
        O\left(m\sqrt{Knn_{\rm sol}}\log\left(\frac{Kn}{\delta}\right)\right) \label{eq:compNaiveOseq}
    \end{equation}
    queries to $O_{\rm seq}$ and $O_{\rm PWM}$, and
    \begin{equation}
        O\left(\sqrt{Knn_{\rm sol}}\log\left(\frac{Kn}{\delta}\right)\right) \label{eq:compNaiveOP}
    \end{equation}
    queries to $O_{\mathcal{P}}$ with $\mathcal{P}$ being some subsets in $\mathcal{P}_{\rm all}$.

    \item If $\mathcal{P}_{\rm sol}$ is empty, the algorithm certainly outputs the message ``no match", making
    \begin{equation}
        O\left(m\sqrt{Kn}\log\left(\frac{Kn}{\delta}\right)\right) \label{eq:compNaiveOseqNo}
    \end{equation}
    queries to $O_{\rm seq}$ and $O_{\rm PWM}$, and
    \begin{equation}
        O\left(\sqrt{Kn}\log\left(\frac{Kn}{\delta}\right)\right) \label{eq:compNaiveOPNo}
    \end{equation}
    queries to $O_{\mathcal{P}}$ with $\mathcal{P}$ being some subsets in $\mathcal{P}_{\rm all}$.
\end{enumerate}

\label{th:NaiveIter}
\end{theorem}

\begin{proof}
First, note that we can perform the following operation for any subset $\mathcal{P}\subset\mathcal{P}_{\rm all}$:
\begin{widetext}
\begin{eqnarray}
    &&\ket{0}\ket{0}\ket{0}\ket{0}\ket{0}\ket{0}\ket{0} \nonumber \\
    &\rightarrow& \frac{1}{\sqrt{Kn}}\sum_{k=0}^{K-1}\sum_{i=0}^{n-1}\ket{k}\ket{i}\ket{0}\ket{0}\ket{0}\ket{0}\ket{0}, \nonumber \\
    &\rightarrow& \frac{1}{\sqrt{Kn}}\sum_{k=0}^{K-1}\sum_{i=0}^{n-1}\ket{k}\ket{i}\ket{w_{k,i}}\ket{0}\ket{0}\ket{0}\ket{0} \nonumber \\
    &\rightarrow& \frac{1}{\sqrt{Kn}}\sum_{k=0}^{K-1}\sum_{i=0}^{n-1}\ket{k}\ket{i}\ket{w_{k,i}}\ket{w_{\rm th}}\ket{0}\ket{0}\ket{0} \nonumber \\
    &\rightarrow& \frac{1}{\sqrt{Kn}}\sum_{k=0}^{K-1}\sum_{i=0}^{n-1}\ket{k}\ket{i}\ket{w_{k,i}}\ket{w_{\rm th}}\left(\mathbbm{1}_{w_{k,i}\ge w_{\rm th}}\ket{1}+\mathbbm{1}_{w_{k,i}< w_{\rm th}}\ket{0}\right)\ket{0}\ket{0} \nonumber \\
    &\rightarrow& \frac{1}{\sqrt{Kn}}\sum_{k=0}^{K-1}\sum_{i=0}^{n-1}\ket{k}\ket{i}\ket{w_{k,i}}\ket{w_{\rm th}}\left(\mathbbm{1}_{w_{k,i}\ge w_{\rm th}}\ket{1}+\mathbbm{1}_{w_{k,i}< w_{\rm th}}\ket{0}\right)\left(\mathbbm{1}_{(k,i)\in\mathcal{P}}\ket{0}+\mathbbm{1}_{(k,i)\notin\mathcal{P}}\ket{1}\right)\ket{0} \nonumber \\
    &\rightarrow& \frac{1}{\sqrt{Kn}}\sum_{k=0}^{K-1}\sum_{i=0}^{n-1}\ket{k}\ket{i}\ket{w_{k,i}}\ket{w_{\rm th}} \otimes \nonumber \\
    &&\qquad\qquad\qquad\left(\mathbbm{1}_{w_{k,i}\ge w_{\rm th} \ \wedge \ (k,i)\notin\mathcal{P}}\ket{1}\ket{1}\ket{1}+ \mathbbm{1}_{w_{k,i}\ge w_{\rm th} \ \wedge \ (k,i)\in\mathcal{P}}\ket{1}\ket{0}\ket{0}+ \mathbbm{1}_{w_{k,i}< w_{\rm th} \ \wedge \ (k,i)\notin\mathcal{P}}\ket{0}\ket{1}\ket{0}+ \mathbbm{1}_{w_{k,i}< w_{\rm th} \ \wedge \ (k,i)\in\mathcal{P}}\ket{0}\ket{0}\ket{0}\right)\nonumber \\
    &=:& \sqrt{\frac{\left|\mathcal{P}_{\rm sol}\cap\overline{\mathcal{P}}\right|}{Kn}}\ket{\psi_{\mathcal{P},1}}\ket{1}+\sqrt{\frac{\left|\overline{\mathcal{P}_{\rm sol}}\cup\mathcal{P}\right|}{Kn}}\ket{\psi_{\mathcal{P},0}}\ket{0}=:\ket{\Psi_\mathcal{P}}
    \label{eq:OPscit}
\end{eqnarray}
\end{widetext}
Here, the seven kets correspond to the seven quantum registers, among which the first four ones have a sufficient number of qubits and the last three ones are single-qubit.
The complement of a set is determined with the universal set being $\mathcal{P}_{\rm all}$.
\begin{equation}
    \ket{\psi_{\mathcal{P},1}}:=\frac{1}{\sqrt{\left|\mathcal{P}_{\rm sol}\cap\overline{\mathcal{P}}\right|}}\sum_{(k,i)\in \mathcal{P}_{\rm sol}\cap\overline{\mathcal{P}}}\ket{k}\ket{i}\ket{w_{k,i}}\ket{w_{\rm th}}\ket{1}\ket{1}
\end{equation}
is the quantum state on the system consisting of the first to sixth registers, and $\ket{\psi_{\mathcal{P},0}}$ is another state on the same system.
In Eq. (\ref{eq:OPscit}), $O^{\rm EqPr}_K$ and $O^{\rm EqPr}_n$ are used at the first arrow.
At the second arrow, $O_{\rm sc,it}$ is used with the register $R_3,...,R_6$ in Procedure \ref{proc:naive} undisplayed.
At the third arrow, we just set $w_{\rm th}$ on the fourth register.
The fourth and fifth transformations are done by $O_{\rm comp}$ and $O_{\mathcal{P}}$, respectively.
Then, the last transformation is done by a Toffoli gate on the last three qubits.
We denote by $\tilde{O}^{\mathcal{P}}_{\rm sc,it}$ the oracle for the operation in Eq. (\ref{eq:OPscit}).
Note that $\tilde{O}^{\mathcal{P}}_{\rm sc,it}$ contains one call to $O_\mathcal{P}$ and $m$ calls to $O_{\rm seq}$ and $O_{\rm PWM}$, since $O_{\rm sc,it}$ makes $O(m)$ calls to them.

Then, the naive iteration method is presented as Algorithm \ref{alg:Naive}.

\ \\

\SetAlgorithmName{\algorithmsname}{}{}
\begin{algorithm}[H]
    \KwIn{$\delta\in(0,1)$}

    Set $\mathcal{P}_{\rm temp}=\emptyset$ and ${\rm StopFlg}=0$\\
    
    \While{${\rm StopFlg}=0$}{
        Perform $\proc{QAA}\left(\tilde{O}^{\mathcal{P}_{\rm temp}}_{\rm sc,it},\frac{1}{Kn},\frac{\delta}{Kn}\right)$.\\

        \uIf{The output message is ``success"}{
            Measure the first and second registers in the output state $\ket{\psi_{\mathcal{P}_{\rm temp},1}}$.
            \\

            Add the measurement outcome $(k,i)$ to $\mathcal{P}_{\rm temp}$.\\
        }
        \Else{
            Set ${\rm StopFlg}=1$
        }
    }

    \uIf{$\mathcal{P}_{\rm temp}\ne \emptyset$}{
        Output $\mathcal{P}_{\rm temp}$.
    }
    \Else{
        Output "no match".
    }
    
    \caption{the naive iteration method}\label{alg:Naive}
\end{algorithm}

\ \\

Let us consider the case that $n_{\rm sol}\ge1$.
In this case, we run QAA repeatedly, and, if each run finishes with the message ``success", we obtain an element in $\mathcal{P}_{\rm sol}\setminus\mathcal{P}_{\rm temp}$, that is, an element in $\mathcal{P}_{\rm sol}$ that has not been obtained in the previous runs yet.
Thus, if QAAs finish with ``success" $n_{\rm sol}$ times in a row, we obtain all the $n_{\rm sol}$ elements in $\mathcal{P}_{\rm sol}$.
This happens with probability at least
\begin{equation}
    \left(1-\frac{\delta}{Kn}\right)^{n_{\rm sol}}\ge \left(1-\frac{\delta}{Kn}\right)^{Kn}\ge 1-\delta,
\end{equation}
since each QAA finishes with the message ``success" with probability at least $1-\frac{\delta}{Kn}$, according to Theorem \ref{th:QAA}.
In the $l$th QAA, at which the $n_{\rm sol}-l+1$ elements in $\mathcal{P}_{\rm sol}$ remain not to be obtained, the number of the queries to $\tilde{O}^{\mathcal{P}_{\rm temp}}_{\rm sc,it}$ is
\begin{equation}
    O\left(\sqrt{\frac{Kn}{n_{\rm sol}-l+1}}\log\left(\frac{Kn}{\delta}\right)\right) \label{eq:OIterQAA}
\end{equation}
according to Theorem \ref{th:QAA}, since the amplitude of $\ket{\psi_{\mathcal{P}_{\rm temp},1}}$ in $\Ket{\Psi_{\mathcal{P}_{\rm temp}}}$ is
\begin{equation}
    \sqrt{\frac{\left|\mathcal{P}_{\rm sol}\cap \overline{\mathcal{P}_{\rm temp}}\right|}{Kn}}=\sqrt{\frac{n_{\rm sol}-l+1}{Kn}}.
\end{equation}
Thus, in this step, the number of the queries to $O_{\rm seq}$ and $O_{\rm PWM}$ is
\begin{equation}
    O\left(m\sqrt{\frac{Kn}{n_{\rm sol}-l+1}}\log\left(\frac{Kn}{\delta}\right)\right),
\end{equation}
since $\tilde{O}^{\mathcal{P}_{\rm temp}}_{\rm sc,it}$ contains $O(m)$ calls to them, and that of the queries to $O_{\mathcal{P}}$ is of order (\ref{eq:OIterQAA}), since $\tilde{O}^{\mathcal{P}_{\rm temp}}_{\rm sc,it}$ calls it once with $\mathcal{P}$ being $\mathcal{P}_{\rm temp}$.
Therefore, for $O_{\mathcal{P}}$, the total query number in the series of QAAs is
\begin{equation}
    O\left(\sum_{l=1}^{n_{\rm sol}}\sqrt{\frac{Kn}{n_{\rm sol}-l+1}}\log\left(\frac{Kn}{\delta}\right)\right),
\end{equation}
which turns into Eq. (\ref{eq:compNaiveOP}) by simple algebra, and that for $O_{\rm seq}$ and $O_{\rm PWM}$ is this times $m$, that is, Eq. (\ref{eq:compNaiveOseq}).
After $n_{\rm sol}$ QAAs with ``success", we run another QAA that outputs ``failure" and end the algorithm, since now $\mathcal{P}_{\rm sol}\cap\overline{\mathcal{P}_{\rm temp}}$ is empty and the amplitude of $\ket{\psi_{\mathcal{P}_{\rm temp},1}}$ in $\Ket{\Psi_{\mathcal{P}_{\rm temp}}}$ is 0.
In this last QAA, the query number for $O_{\rm seq}$ and $O_{\rm PWM}$ is of order (\ref{eq:compNaiveOseqNo}) and that for $O_{\mathcal{P}}$ is of order (\ref{eq:compNaiveOPNo}), according to Theorem \ref{th:QAA}, but the total query number in the algorithm remains of order (\ref{eq:compNaiveOseq}) and (\ref{eq:compNaiveOP}).

In the case that $n_{\rm sol}=0$, the first QAA outputs ``failure" and then the algorithm ends.
The query number in this is of same order as that in the last QAA in the case that $n_{\rm sol}>0$, that is, Eqs. (\ref{eq:compNaiveOseqNo}) and (\ref{eq:compNaiveOPNo}).

\end{proof}

Let us comment on the number of qubits used in the naive iteration method.
As we see in Eqs. (\ref{eq:Oscit}) and (\ref{eq:OPscit}), this algorithm uses several quantum registers to represent the indexes for PWMs, positions in a sequence, and positions in a segment.
The sufficient qubit numbers in these types of registers are $O(\log K)$, $O(\log n)$ and $O(\log m)$, respectively, and thus the total qubit number is $O(\log K + \log n + \log m)$, logarithmic on the parameters that characterize the problem.
Besides, the algorithm uses a few registers to represent real numbers such as an entry $M_k(j,s)$ of a PWM and a segment score $w_{k,i}$.
If we take some typical setting for binary representation of real numbers (say, double precision with 64 bits) independently of the problem, a few registers for real numbers amount qubits of order $10^2$, which surpasses the qubits for the indexes for typical values of the parameters $K,n$ and $m$ mentioned in Sec. \ref{sec:PWM}.
In summary, for a typical PWM matching problem, the qubit number used in the naive iteration method is of order $10^2$.
This also applies to the QMCI-based method, which is explained in Sec. \ref{sec:QMCIBaseMethod}.

\subsection{Algorithm \rom{2}: QMCI-based method \label{sec:QMCIBaseMethod}}

Next, we present the QMCI-based method.
It is basically same as the naive iteration method, but it calculate the score of each segment by QMCI.
Although in the naive iteration method we calculate the score of one segment, which is a sum of the $m$ position-wise scores, calling $O_{\rm seq}$ and $O_{\rm PWM}$ $O(m)$ times, this query number can be reduced by QMCI, whose complexity depends on the required accuracy, if we can set it sufficiently loose.

This method is inspired by the algorithm for GW matchied filtering presented in \cite{Miyamoto2022}.
It is the two-fold problem of calculating the quantity called SNR, which is given as the sum of many terms, for each template waveform of the GW signal, and searching the high-SNR templates.
In this regard, it has a same structure as PWM matching, which consists of calculating the scores of the segments and searching the high-score segments.
Therefore, we naturally conceive the idea to apply the algorithm in \cite{Miyamoto2022}, which is a combination of QMCI and QAA, to PWM matching.

As pointed out in \cite{Miyamoto2022}, there is an issue in using QMCI.
The result of QMCI inevitably contains the error, and it can cause the {\it false match}.
That is, even if $w_{k,i}$ for some $(k,i)\in\mathcal{P}_{\rm all}$ is smaller than the threshold $w_{\rm th}$, the estimation of it by QMCI might exceed $w_{\rm th}$ due to the error, and we might misjudge that the $i$th segment matches with $M_k$.
To cope with this, we set the threshold in the similar way to \cite{Miyamoto2022}.
That is, we set the two levels of the threshold, $w_{\rm soft}$ and $w_{\rm soft}$, which have the following meanings:
\begin{itemize}
    \item We want to find $(k,i)\in\mathcal{P}_{\rm all}$ such that $w_{k,i}\ge w_{\rm hard}$ with high probability.

    \item We never want to falsely find $(k,i)$ such that $w_{k,i}< w_{\rm soft}$.

    \item If there are $(k,i)$ such that $w_{\rm soft}\le w_{k,i}< w_{\rm hard}$, it is not necessary but fine to find them.
\end{itemize}
Then, we set the accuracy of QMCI to $\frac{w_{\rm hard}-w_{\rm soft}}{2}$ and take the following policy: if the estimation of $w_{k,i}$ by QMCI is larger than or equal to
\begin{equation}
    w_{\rm mid}:=\frac{w_{\rm soft}+w_{\rm hard}}{2}, \label{eq:wmid}    
\end{equation}
we judge $(k,i)$ as ``matched", and if not, we judge as ``mismatched".
Under this policy, $(k,i)$ is judged as ``matched" if $w_{k,i}\ge w_{\rm hard}$ and ``mismatched" if $w_{k,i}< w_{\rm soft}$ with high probability. 
We will discuss the validity to assume that such two threshold levels are set in Sec. \ref{sec:ScoreTol}. 

Now, let us present the theorem on the QMCI-based method.

\begin{theorem}
Consider Problem \ref{prob:PWMM} under Assumptions \ref{ass:oracle} and \ref{ass:OP}.
Suppose that we are given $\delta\in(0,1)$ and $w_{\rm soft},w_{\rm hard}\in\mathbb{R}$ such that $0<w_{\rm soft}<w_{\rm hard}<m$. 
Define
\begin{equation}
    \mathcal{P}_{\rm hard} := \{(k,i)\in\mathcal{P}_{\rm all} \ | \ w_{k,i}\ge w_{\rm hard}\}
\end{equation}
and
\begin{equation}
    \mathcal{P}_{\rm soft} := \{(k,i)\in\mathcal{P}_{\rm all} \ | w_{k,i} \ge w_{\rm soft}\}.
\end{equation}
Then, there is a quantum algorithm that makes queries to $O_{\rm seq}$, $O_{\rm PWM}$ and $O_{\mathcal{P}}$ with $\mathcal{P}$ being some subsets in $\mathcal{P}_{\rm all}$, and behaves as follows:
\begin{enumerate}
    \renewcommand{\labelenumi}{(\roman{enumi})}
    \item If $n_{\rm hard}:=|\mathcal{P}_{\rm hard}|>0$, the algorithm outputs all the elements in $\mathcal{P}_{\rm hard}$ and 0 or more elements in $\mathcal{P}_{\rm soft}\setminus\mathcal{P}_{\rm hard}$ with probability at least $1-\delta$.
    In the algorithm, $O_{\rm seq}$ and $O_{\rm PWM}$ are called
    \begin{equation}
        O\left(\frac{mn_{\rm soft}\sqrt{Kn}}{w_{\rm hard}-w_{\rm soft}}\log\left(\frac{K^2n^2}{\delta}\right)\log\left(\frac{Kn}{\delta}\right)\right) \label{eq:compQMCIOseq1}
    \end{equation}
    times, and $O_{\mathcal{P}}$ are called
    \begin{equation}
        O\left(n_{\rm soft}\sqrt{Kn}\log\left(\frac{K^2n^2}{\delta}\right)\log\left(\frac{Kn}{\delta}\right)\right)  \label{eq:compQMCIOP1}
    \end{equation}
    times, where $n_{\rm soft}:=|\mathcal{P}_{\rm soft}|$.

    \item If $n_{\rm soft}=0$, the algorithm certainly outputs the message ``no match".
    In the algorithm, $O_{\rm seq}$ and $O_{\rm PWM}$ are called
    \begin{equation}
        O\left(\frac{m\sqrt{Kn}}{w_{\rm hard}-w_{\rm soft}}\log\left(\frac{K^2n^2}{\delta}\right)\log\left(\frac{Kn}{\delta}\right)\right) \label{eq:compQMCIOseq2}
    \end{equation}
    times, and $O_{\mathcal{P}}$ are called
    \begin{equation}
        O\left(\sqrt{Kn}\log\left(\frac{K^2n^2}{\delta}\right)\log\left(\frac{Kn}{\delta}\right)\right) \label{eq:compQMCIOP2}
    \end{equation}
    times.

    \item If $n_{\rm soft}>0$ and $n_{\rm hard}=0$, the algorithm certainly outputs the message ``no match" or 1 or more elements in $\mathcal{P}_{\rm soft}$.
    In the algorithm, the number of queries to $O_{\rm seq}$ and $O_{\rm PWM}$ is of order (\ref{eq:compQMCIOseq1}), and the number of queries to $O_{\mathcal{P}}$ is of order (\ref{eq:compQMCIOP1}).
\end{enumerate}

\label{th:QMCIBased}
\end{theorem}

\begin{proof}

First, note that, for any $k\in[K]_0$, $i\in[n-m+1]_0$ and $j\in[m]_0$, we can perform the following operation:
\begin{eqnarray}
&&\ket{k}\ket{i}\ket{j}\ket{0}\ket{0}\ket{0} \nonumber \\
&\rightarrow& \ket{k}\ket{i}\ket{j}\ket{i+j}\ket{0}\ket{0} \nonumber \\
&\rightarrow& \ket{k}\ket{i}\ket{j}\ket{i+j}\ket{s_{i+j}}\ket{0} \nonumber \\
&\rightarrow& \ket{k}\ket{i}\ket{j}\ket{i+j}\ket{s_{i+j}}\ket{M_k(j,s_{i+j})}, \label{eq:Oscone}
\end{eqnarray}
where $O_{\rm add}$, $O_{\rm seq}$ and $O_{\rm PWM}$ are used at the first, second and third arrows, respectively.
We denote by $O_{\rm sc,one}$ the oracle for the above operation.
According to Theorem \ref{th:QMCI}, for any $\epsilon,\delta\in(0,1)$, we use $O_{\rm sc,one}$ $O(\epsilon^{-1}\log(\delta^{-1}))$ times to construct the oracle $O^{\rm sc,QMCI}_{\epsilon,\delta}$ that acts as
\begin{equation}
    O^{\rm sc,QMCI}_{\epsilon,\delta}\ket{k}\ket{i}\ket{0}=\ket{k}\ket{i}\sum_{y\in \mathcal{Y}_{k,i}} \alpha^{k,i}_y\ket{y}.
\end{equation}
Here, $\mathcal{Y}_{k,i}$ is a finite set of real numbers that includes a subset $\tilde{\mathcal{Y}}_{k,i}$ consisting of $\epsilon$-approximations of $\frac{w_{k,i}}{m}$, and $\{\alpha^{k,i}_y\}_{y\in\mathcal{Y}_{k,i}}$ are complex numbers satisfying
\begin{equation}
\sum_{\tilde{y}\in\tilde{\mathcal{Y}}_{k,i}}|\alpha^{k,i}_{\tilde{y}}|^2\ge 1-\delta.
\end{equation}
Furthermore, as Eq. (\ref{eq:OPscit}), we can construct $\tilde{O}^{\rm sc,QMCI}_{\epsilon,\delta,\mathcal{P}}$ that acts on the seven-register system as
\begin{widetext}
\begin{eqnarray}
    &&\tilde{O}^{\rm sc,QMCI}_{\epsilon,\delta,\mathcal{P}}\ket{0}\ket{0}\ket{0}\ket{0}\ket{0}\ket{0}\ket{0}\nonumber \\
    &=& \frac{1}{\sqrt{Kn}}\sum_{k=0}^{K-1}\sum_{i=0}^{n-1}\sum_{y\in \mathcal{Y}_{k,i}} \alpha^{k,i}_y\ket{k}\ket{i}\ket{y}\Ket{\frac{w_{\rm mid}}{m}}\left(\mathbbm{1}_{y\ge \frac{w_{\rm mid}}{m} \ \wedge \ (k,i)\notin\mathcal{P}}\ket{1}\ket{1}\ket{1}+ \mathbbm{1}_{y\ge \frac{w_{\rm mid}}{m} \ \wedge \ (k,i)\in\mathcal{P}}\ket{1}\ket{0}\ket{0}+ \right. \nonumber \\
    && \ \quad\qquad\qquad\qquad\qquad\qquad\qquad\qquad \left.\mathbbm{1}_{y< \frac{w_{\rm mid}}{m} \ \wedge \ (k,i)\notin\mathcal{P}}\ket{0}\ket{1}\ket{0}+\mathbbm{1}_{y< \frac{w_{\rm mid}}{m} \ \wedge \ (k,i)\in\mathcal{P}}\ket{0}\ket{0}\ket{0}\right) \nonumber \\
    &=:& \beta_{\mathcal{P},1}\ket{\xi_{\mathcal{P},1}}\ket{1}+\beta_{\mathcal{P},0}\ket{\xi_{\mathcal{P},0}}\ket{0}=:\ket{\Xi_\mathcal{P}}
\end{eqnarray}
\end{widetext}
for any subset $\mathcal{P}\subset\mathcal{P}_{\rm all}$, using $O^{\rm sc,QMCI}_{\epsilon,\delta}$, $O_\mathcal{P}$ and some arithmetic oracles.
Here,
\begin{eqnarray}
    &&\ket{\xi_{\mathcal{P},1}}:= \frac{1}{\sqrt{\sum_{(k,i)\in\overline{\mathcal{P}}}\sum_{\substack{y\in \mathcal{Y}_{k,i} \\ y\ge w_{\rm mid}/m}}|\alpha^{k,i}_y|^2}} \times \nonumber \\
    &&\qquad\qquad\sum_{(k,i)\in\overline{\mathcal{P}}}\sum_{\substack{y\in \mathcal{Y}_{k,i} \\ y\ge w_{\rm mid}/m}} \alpha^{k,i}_y\ket{k}\ket{i}\ket{y}\Ket{\frac{w_{\rm mid}}{m}}\ket{1}\ket{1} 
\end{eqnarray}
is the quantum states on the first six register, and $\ket{\xi_{\mathcal{P},0}}$ is another state on the same system.
\begin{equation}
    \beta_{\mathcal{P},1}=\sqrt{\frac{\sum_{(k,i)\in\overline{\mathcal{P}}}\sum_{\substack{y\in \mathcal{Y}_{k,i} \\ y\ge w_{\rm mid}/m}}|\alpha^{k,i}_y|^2}{Kn}},
\end{equation}
and $\beta_{\mathcal{P},0}$ is another complex number satisfying $|\beta_{\mathcal{P},0}|^2+|\beta_{\mathcal{P},1}|^2=1$.
Note that $\tilde{O}^{\rm sc,QMCI}_{\epsilon,\delta,\mathcal{P}}$ uses $O^{\rm sc,QMCI}_{\epsilon,\delta}$ once, and thus $O_{\rm sc,one}$ $O(\epsilon^{-1}\log(\delta^{-1}))$ times.
Since $O_{\rm sc,one}$ calls $O_{\rm seq}$ and $O_{\rm PWM}$ once each, $\tilde{O}^{\rm sc,QMCI}_{\epsilon,\delta,\mathcal{P}}$ calls them $O(\epsilon^{-1}\log(\delta^{-1}))$ times, consequently.
Also note that $\tilde{O}^{\rm sc,QMCI}_{\epsilon,\delta,\mathcal{P}}$ uses $O_{\mathcal{P}}$ once.

Then, we present the QMCI-based method as Algorithm \ref{alg:QMCIbased}.

\SetAlgorithmName{\algorithmsname}{}{}
\begin{algorithm}[H]
    \KwIn{
        \begin{itemize}
        \item $\delta\in(0,1)$
        \item $w_{\rm soft}$ and $w_{\rm hard}$ such that $0<w_{\rm soft}<w_{\rm hard}<m$
        \end{itemize}
    }

    Set $\mathcal{P}_{\rm temp}=\emptyset$, ${\rm StopFlg}=0$, $\delta^\prime=\frac{\delta}{4K^2n^2}$, $\delta^{\prime\prime}=\frac{\delta}{2Kn}$ and 
    \begin{equation}
        \epsilon^\prime=\frac{w_{\rm hard}-w_{\rm soft}}{2m}. \label{eq:epspr}
    \end{equation}
    \\
    
    \While{${\rm StopFlg}=0$}{
        Perform $\proc{QAA}\left(\tilde{O}^{\rm sc,QMCI}_{\epsilon^\prime,\delta^\prime,\mathcal{P}_{\rm temp}},\frac{1}{2Kn},\delta^{\prime\prime}\right)$.\\

        \uIf{The output message is ``success"}{
            Measure the first and second registers in the output state $\ket{\xi_{\mathcal{P}_{\rm temp},1}}$ and let the outcome be $(k,i)$.
            \\

            Calculate $w_{k,i}$ classically and let the result be $w^{\rm cl}_{k,i}$.

            \uIf{$w^{\rm cl}_{k,i}\ge w_{\rm soft}$}{
            Add $(k,i)$ to $\mathcal{P}_{\rm temp}$.
            }\Else{
                Set ${\rm StopFlg}=1$
            }   
        }
        \Else{
            Set ${\rm StopFlg}=1$
        }
    }

    \uIf{$\mathcal{P}_{\rm temp}\ne \emptyset$}{
        Output $\mathcal{P}_{\rm temp}$.
    }
    \Else{
        Output "no match".
    }
    \caption{The QMCI-based method}\label{alg:QMCIbased}
\end{algorithm}

\ \\

Then, let us consider the behavior of this algorithm in the following cases.

\ \\

\noindent (i) $n_{\rm hard}>0$ \\

For any $(k,i)\in\mathcal{P}_{\rm hard}$,
\begin{equation}
    \left|y-\frac{w_{k,i}}{m}\right|\le \epsilon^\prime \Rightarrow y\ge \frac{w_{\rm mid}}{m}
\end{equation}
holds for any $y\in\mathbb{R}$ under the definitions (\ref{eq:wmid}) and (\ref{eq:epspr}), and thus
\begin{equation}
    \sum_{\substack{y\in \mathcal{Y}_{k,i} \\ y\ge w_{\rm mid}/m}}|\alpha^{k,i}_{y}|^2\ge \sum_{\substack{y\in \mathcal{Y}_{k,i} \\ |y-w_{k,i}|\le\epsilon^\prime}}|\alpha^{k,i}_{y}|^2\ge 1-\delta^\prime\ge \frac{1}{2} \label{eq:alphasq}
\end{equation}
holds.
This means that, if $\mathcal{P}_{\rm hard}\cap\overline{\mathcal{P}_{\rm temp}}\ne\emptyset$,
\begin{eqnarray}
    |\beta_{\mathcal{P}_{\rm temp},1}|^2&=&\frac{\sum_{(k,i)\in\overline{\tilde{\mathcal{P}}}_{\rm temp}}\sum_{\substack{y\in \mathcal{Y}_{k,i} \\ y\ge w_{\rm mid}/m}}|\alpha^{k,i}_y|^2}{Kn} \nonumber \\
    &\ge& \frac{\sum_{(k,i)\in\mathcal{P}_{\rm hard}\cap\overline{\mathcal{P}_{\rm temp}}}\sum_{\substack{y\in \mathcal{Y}_{k,i} \\ y\ge w_{\rm mid}/m}}|\alpha^{k,i}_y|^2}{Kn}
    \nonumber \\
    &\ge& \sum_{(k,i)\in\mathcal{P}_{\rm hard}\cap\overline{\mathcal{P}_{\rm temp}}}\frac{1}{2Kn} \nonumber \\
    &\ge& \frac{1}{2Kn}, \label{eq:betasq}
\end{eqnarray}
and thus $\proc{QAA}\left(\tilde{O}^{\rm sc,QMCI}_{\epsilon^\prime,\delta^\prime,\mathcal{P}_{\rm temp}},\frac{1}{2Kn},\delta^{\prime\prime}\right)$ outputs $\ket{\xi_{\mathcal{P}_{\rm temp},1}}$ with probability at least $1-\delta^{\prime\prime}=1-\frac{\delta}{2Kn}$.

On the other hand, for $(k,i)\notin\mathcal{P}_{\rm soft}$,
\begin{equation}
        y\ge \frac{w_{\rm mid}}{m}\Rightarrow\left|y-\frac{w_{k,i}}{m}\right|>\epsilon^\prime\Rightarrow y\notin\tilde{\mathcal{Y}}_{k,i}
\end{equation}
holds for any $y\in\mathbb{R}$, and thus
\begin{equation}
    \sum_{\substack{y\in \mathcal{Y}_{k,i} \\ y\ge w_{\rm mid}/m}}|\alpha^{k,i}_{y}|^2 = \sum_{\substack{y\in \mathcal{Y}_{k,i}\setminus\tilde{\mathcal{Y}}_{k,i} \\ y\ge w_{\rm mid}/m}}|\alpha^{k,i}_{y}|^2\le \sum_{y\in \mathcal{Y}_{k,i}\setminus\tilde{\mathcal{Y}}_{k,i}}|\alpha^{k,i}_{y}|^2 < \delta^\prime \label{eq:alphaSqSum}
\end{equation}
holds.
From this, the probability that we obtain $(k,i)\in\mathcal{P}_{\rm soft}$ in measuring the first two registers in $\ket{\xi_{\mathcal{P}_{\rm temp},1}}$ is evaluated as
\begin{eqnarray}
&&\frac{\sum_{(k,i)\in\mathcal{P}_{\rm soft}\cap\overline{\mathcal{P}_{\rm temp}}}\sum_{\substack{y\in \mathcal{Y}_{k,i} \\ y\ge w_{\rm mid}/m}}|\alpha^{k,i}_y|^2}{\sum_{(k,i)\in\overline{\mathcal{P}_{\rm temp}}}\sum_{\substack{y\in \mathcal{Y}_{k,i} \\ y\ge w_{\rm mid}/m}}|\alpha^{k,i}_y|^2} \nonumber \\
&=& 1-\frac{\sum_{(k,i)\in\overline{\mathcal{P}_{\rm soft}}\cap\overline{\mathcal{P}_{\rm temp}}}\sum_{\substack{y\in \mathcal{Y}_{k,i} \\ y\ge w_{\rm mid}/m}}|\alpha^{k,i}_y|^2}{\sum_{(k,i)\in\overline{\mathcal{P}_{\rm temp}}}\sum_{\substack{y\in \mathcal{Y}_{k,i} \\ y\ge w_{\rm mid}/m}}|\alpha^{k,i}_y|^2} \nonumber \\
&\ge& 1-2\sum_{(k,i)\in\overline{\mathcal{P}_{\rm soft}}\cap\overline{\mathcal{P}_{\rm temp}}}\delta^\prime \nonumber \\
&\ge& 1-2Kn\delta^\prime \nonumber \\
&=& 1-\frac{\delta}{2Kn}.
\end{eqnarray}
At the first inequality, we used Eq. (\ref{eq:alphaSqSum}) and
\begin{equation}
    \sum_{(k,i)\in\overline{\mathcal{P}_{\rm temp}}}\sum_{\substack{y\in \mathcal{Y}_{k,i} \\ y\ge w_{\rm mid}/m}}|\alpha^{k,i}_y|^2=Kn|\beta_{\mathcal{P}_{\rm temp},1}|^2\ge\frac{1}{2},
\end{equation}
which follows from Eq. (\ref{eq:betasq}). 

Combining the above discussions, we see that, if $\mathcal{P}_{\rm hard}\cap\overline{\mathcal{P}_{\rm temp}}\ne\emptyset$, we obtain an element in $\mathcal{P}_{\rm soft}\cap\overline{\mathcal{P}_{\rm temp}}$ by $\proc{QAA}\left(\tilde{O}^{\rm sc,QMCI}_{\epsilon^\prime,\delta^\prime,\mathcal{P}_{\rm temp}},\frac{1}{2Kn},\delta^{\prime\prime}\right)$ and the subsequent measurement on $\ket{\xi_{\mathcal{P}_{\rm temp},1}}$ with probability at least $\left(1-\frac{\delta}{2Kn}\right)^2\ge 1-\frac{\delta}{Kn}$.
Therefore, with some probability, the following happens: we successively obtain elements in $\mathcal{P}_{\rm soft}$ in loop 2-12 in Algorithm \ref{alg:QMCIbased}, until we get all the elements in $\mathcal{P}_{\rm hard}$.
Since the number of loops is at most $n_{\rm soft}$, the probability that this happens is at least
\begin{equation}
    \left(1-\frac{\delta}{Kn}\right)^{n_{\rm soft}} \ge \left(1-\frac{\delta}{Kn}\right)^{Kn} \ge 1-\delta.
\end{equation}
We can evaluate the query complexity in this loop as Eqs. (\ref{eq:compQMCIOseq1}) and (\ref{eq:compQMCIOP1}) under the current setting of $\epsilon^\prime,\delta^\prime$ and $\delta^{\prime\prime}$, recalling that $\proc{QAA}\left(\tilde{O}^{\rm sc,QMCI}_{\epsilon^\prime,\delta^\prime,\mathcal{P}_{\rm temp}},\frac{1}{2Kn},\delta^{\prime\prime}\right)$ calls $\tilde{O}^{\rm sc,QMCI}_{\epsilon^\prime,\delta^\prime,\mathcal{P}_{\rm temp}}$ $O\left(\sqrt{Kn}\log\left(\frac{1}{\delta^{\prime\prime}}\right)\right)$ times and that $\tilde{O}^{\rm sc,QMCI}_{\epsilon^\prime,\delta^\prime,\mathcal{P}_{\rm temp}}$ calls $O_{\rm seq}$ and $O_{\rm PWM}$ $O\left(\frac{1}{\epsilon^\prime}\log\left(\frac{1}{\delta^\prime}\right)\right)$ times and $O_\mathcal{P}$ $O(1)$ times.

\ \\

\noindent (ii) $n_{\rm soft}>0$ \\

With certainty, the first run of QAA outputs the message ``failure" or, even if not, $w_{k,i}$ classically calculated at step 6 in Algorithm \ref{alg:QMCIbased} is smaller than $w_{\rm soft}$, since every $(k,i)\in\mathcal{P}_{\rm all}$ yields $w_{k,i}<w_{\rm soft}$ in this case.
Therefore, the algorithm certainly ends outputting ``no match", with QAA run only once.
The query complexity of this is evaluated as Eqs. (\ref{eq:compQMCIOseq2}) and (\ref{eq:compQMCIOP2}).

\ \\

\noindent (iii) $n_{\rm hard}=0$ and $n_{\rm soft}>0$ \\

The algorithm ends with only one QAA that outputs ``failure" or $(k,i)\in\mathcal{P}_{\rm all}$ such that $w_{k,i}<w_{\rm soft}$, or QAA runs some times outputting $(k,i)\in\mathcal{P}_{\rm soft}$.
The number of QAA loops is at most $n_{\rm soft}$, and thus the query complexity is evaluated as Eqs. (\ref{eq:compQMCIOseq1}) and (\ref{eq:compQMCIOP1}) similarly to the case (i).

\end{proof}

Let us make some comments.
First, note that, although Eqs. (\ref{eq:compNaiveOseq}) and (\ref{eq:compNaiveOseqNo}) scales with $n_{\rm sol}$ as $O(\sqrt{n_{\rm sol}})$, Eqs. (\ref{eq:compQMCIOseq1}) and (\ref{eq:compQMCIOP1}) scales linearly with $n_{\rm soft}$, which means that the QMCI-based method has the worse scaling with respect to the number of matches.
This difference can be understood as follows.
In the naive iteration method, among the computational basis states contained in the state $\ket{\Psi_{\mathcal{P}_{\rm temp}}}$, those with the last qubit taking $\ket{1}$ are $\ket{k}\ket{i}\ket{w_{k,i}}\ket{w_{\rm th}}\ket{1}\ket{1}\ket{1}$ with $(k,i)\in\mathcal{P}_{\rm sol}\cap\overline{\mathcal{P}_{\rm temp}}$, each of which having the amplitude $\sqrt{\frac{1}{Kn}}$.
They constitute $\ket{\psi_{\mathcal{P}_{\rm temp},1}}\ket{1}$, the target state of QAA, whose amplitude decreases as $\sqrt{\frac{n_{\rm sol}}{Kn}},\sqrt{\frac{n_{\rm sol}-1}{Kn}},...,\sqrt{\frac{1}{Kn}}$ in the QAA loop, and this leads to the evaluation of the total complexity in Eqs. (\ref{eq:compQMCIOseq1}) and (\ref{eq:compQMCIOP1}).
On the other hand, in the QMCI-based method, among the computational basis states contained in the state $\ket{\Xi_{\mathcal{P}_{\rm temp}}}$, those with the last qubit taking $\ket{1}$ are $\ket{k}\ket{i}\ket{y}\Ket{\frac{w_{\rm th}}{m}}\ket{1}\ket{1}\ket{1}$ with $(k,i)$ being any elements in $\overline{\mathcal{P}_{\rm temp}}$, although those for $(k,i)\in\mathcal{P}_{\rm soft}$ constitute the most part of $\ket{\Xi_{\mathcal{P}_{\rm temp}}}$ in terms of the squared amplitude.
When we write $\ket{\Xi_{\mathcal{P}_{\rm temp}}}$ as
\begin{equation}
    \ket{\Xi_{\mathcal{P}_{\rm temp}}}=\frac{1}{\sqrt{Kn}}\sum_{(k,i)\in\overline{\mathcal{P}_{\rm temp}}}\gamma_{k,i}\ket{\tilde{\xi}_{\mathcal{P}_{\rm temp},1;k,i}}+\beta_{\mathcal{P}_{\rm temp},0}\ket{\xi_{\mathcal{P}_{\rm temp},0}}\ket{0}
    \end{equation}
with
\begin{eqnarray}
    \ket{\tilde{\xi}_{\mathcal{P}_{\rm temp},1;k,i}}&:=&\frac{1}{\gamma_{k,i}}\ket{k}\ket{i}\sum_{\substack{y\in \mathcal{Y}_{k,i} \\ y\ge w_{\rm mid}/m}}\alpha^{k,i}_y\ket{y}\Ket{\frac{w_{\rm th}}{m}}\ket{1}\ket{1}\ket{1}, \nonumber \\
    \gamma_{k,i}&:=&\sqrt{\sum_{\substack{y\in \mathcal{Y}_{k,i} \\ y\ge w_{\rm mid}/m}}\left|\alpha^{k,i}_y\right|^2},
\end{eqnarray}
the squared amplitude of $\ket{\tilde{\xi}_{\mathcal{P}_{\rm temp},1;k,i}}$, $\frac{|\gamma_{k,i}|^2}{Kn}$, is at least $\frac{1}{2Kn}$ for $(k,i)\in\mathcal{P}_{\rm hard}$ as we see from Eq. (\ref{eq:alphasq}), but that for $(k,i)\in\mathcal{P}_{\rm soft}\setminus\mathcal{P}_{\rm hard}$ can be much smaller than it.
Nevertheless, the squared amplitudes of the states $\ket{\tilde{\xi}_{\mathcal{P}_{\rm temp},1;k,i}}$ for $(k,i)\in\mathcal{P}_{\rm soft}\setminus\mathcal{P}_{\rm hard}$ can pile up to the value comparable with $\frac{1}{2Kn}$.
In such a situation, it is possible that, in the QAA loop, $\proc{QAA}\left(\tilde{O}^{\rm sc,QMCI}_{\epsilon^\prime,\delta^\prime,\mathcal{P}_{\rm temp}},\frac{1}{2Kn},\delta^{\prime\prime}\right)$ continues to output $\ket{\xi_{\mathcal{P}_{\rm temp},1}}$ and we continue to get $(k,i)\in\mathcal{P}_{\rm soft}$, until we get $O(n_{\rm soft})$ elements in $\mathcal{P}_{\rm soft}$ and the squared amplitude $|\beta_{\mathcal{P}_{\rm temp},1}|^2$ of $\ket{\xi_{\mathcal{P}_{\rm temp},1}}\ket{1}$ in $\ket{\Xi_{\mathcal{P}_{\rm temp}}}$ decreases below $\frac{1}{2Kn}$.
When this happens, the query complexity becomes comparable with the bounds (\ref{eq:compQMCIOseq1}) and (\ref{eq:compQMCIOP1}).

Second, seemingly, the bounds on the number of queries to $O_{\rm seq}$ and $O_{\rm PWM}$ in Eqs. (\ref{eq:compQMCIOseq1}) and (\ref{eq:compQMCIOseq2}) linearly scale with $m$, which is similar to Eq. (\ref{eq:compNaiveOseq}) and (\ref{eq:compNaiveOseqNo}) and makes us consider that there is no speedup with respect to $m$ compared to the naive iteration method.
However, if we can set $w_{\rm soft}$ and $w_{\rm hard}$ with larger difference for larger $m$, the dependence of the bounds in Eqs. (\ref{eq:compQMCIOseq1}) and (\ref{eq:compQMCIOseq2}) on $m$ becomes milder than linear.
This seems reasonable because, naively thinking, the typical value of the segment score, which is the sum of $m$ terms, becomes larger for larger $m$, and so do $w_{\rm soft}$, $w_{\rm hard}$ and their difference.
In fact, in Sec. \ref{sec:ScoreTol}, we argue that it is reasonable to take $w_{\rm hard}$ and $w_{\rm soft}$ so that
\begin{equation}
w_{\rm hard}-w_{\rm soft}=\Omega(\sqrt{m}), \label{eq:whardwsoftdiff}
\end{equation}
from which Eqs. (\ref{eq:compQMCIOseq1}) and (\ref{eq:compQMCIOseq2}) turns into
\begin{equation}
    O\left(n_{\rm soft}\sqrt{Knm}\log\left(\frac{K^2n^2}{\delta}\right)\log\left(\frac{Kn}{\delta}\right)\right) \label{eq:compQMCIOseq3}
\end{equation}
and
\begin{equation}
    O\left(\sqrt{Knm}\log\left(\frac{K^2n^2}{\delta}\right)\log\left(\frac{Kn}{\delta}\right)\right),  \label{eq:compQMCIOseq4}
\end{equation}
respectively.
If so, the QMCI-based method can be beneficial compared to the naive iteration method for small $n_{\rm soft}$ and large $m$, that is, in the case that there is a small number of matches and the sequence motif length is large.

\section{Discussion \label{sec:dis}}

\subsection{Implementations of the oracles with QRAMs and the cost to prepare them \label{sec:QRAM}}

Now, we consider how to implement the oracles $O_{\rm seq}$, $O_{\rm PWM}$ and $O_{\mathcal{P}}$, which we have simply assumed are implementable so far.

It seems that, in order to realize the quantum access to the elements in the sequence $S$ like Eq. (\ref{eq:Oseq}), we need to use a QRAM \cite{Giovannetti2008QRAM}.
Although some difficulties in constructing it in reality have been pointed out \cite{Arunachalam_2015}, it is the very device that provides the access to the indexed data in superposition in $O(\log N)$ time with respect to $N$ the number of the data points.
Of course, preparing a QRAM, that is, registering the $N$ data points into the QRAM requires $O(N)$ time.
To prepare $O_{\rm seq}$, we need $O(n)$ time.

We can use a QRAM also for $O_{\rm PWM}$ in Eq. (\ref{eq:OPWM}).
Although the indexes are now three-fold, $(k,i,a)$, it is straightforward to combine them and regard it as an integer.
Preparing this takes $O(m|\mathcal{A}|K)$ time, which is expected to be much shorter than $O(N)$ in usual situations.

We can also construct $O_{\mathcal{P}}$, especially $O_{\mathcal{P}_{\rm temp}}$, using a QRAM.
Naively thinking, we can do this by registering 0 or 1, which represents $(k,i)\in\mathcal{P}$ or not, for every $(k,i)\in\mathcal{P}_{\rm all}$.
However, this takes $O(nK)$ time, which exceeds $O(nm)$ time for the classical exhaustive search if $K>m$.
Therefore, we adopt the following approach that takes the shorter time for QRAM preparation.
First, we plausibly assume that the number of the matched PWMs at every position $i$ in the sequence $S$ is at most $\kappa$, which is $O(1)$.
Then, we prepare the QRAM $\tilde{O}_{\mathcal{P}_{\rm temp}}$ that outputs $\kappa$ indices $k_{i,1},...,k_{i,\kappa}\in[K]_0$ such that $(k_{i,1},i),...,(k_{i,\kappa},i)\in\mathcal{P}_{\rm temp}$ for each $i\in[n]_0$:
\begin{equation}
    \tilde{O}_{\mathcal{P}_{\rm temp}}\ket{i}\underbrace{\ket{0}\cdots\ket{0}}_\kappa=\ket{i}\ket{k_{i,1}}\cdots\ket{k_{i,\kappa}}.
\end{equation}
If $\kappa^\prime$ the number of such indices is smaller than $\kappa$, we set $k_{i,\kappa^\prime+1},...,k_{i,\kappa}$ to some dummy number (say, $-1$) not contained in $[K]_0$.
Using this, we can perform the following operation for any $(k,i)\in\mathcal{P}_{\rm all}$:
\begin{eqnarray}
    &&\ket{k}\ket{i}\underbrace{\ket{0}\cdots\ket{0}}_\kappa \underbrace{\ket{0}\cdots\ket{0}}_\kappa\ket{0} \nonumber \\
    &\rightarrow&\ket{k}\ket{i}\ket{k_{i,1}}\cdots\ket{k_{i,\kappa}} \ket{0}\cdots\ket{0}\ket{0} \nonumber \\
    &\rightarrow&\ket{k}\ket{i}\ket{k_{i,1}}\cdots\ket{k_{i,\kappa}} \Ket{\mathbbm{1}_{k\ne k_{i,1}}}\cdots\Ket{\mathbbm{1}_{k \ne 
 k_{i,\kappa}}}\ket{0} \nonumber \\
    &\rightarrow&\ket{k}\ket{i}\ket{k_{i,1}}\cdots\ket{k_{i,\kappa}} \Ket{\mathbbm{1}_{k\ne k_{i,1}}}\cdots\Ket{\mathbbm{1}_{k \ne 
 k_{i,\kappa}}}\Ket{\mathbbm{1}_{k \ne 
 k_{i,1} \wedge\cdots\wedge k \ne 
 k_{i,\kappa}}}. \nonumber \\
 && \label{eq:OPImpl}
\end{eqnarray}
Here, the first to $(\kappa+2)$th kets correspond to registers with the sufficient number of qubits and the other kets correspond to the single qubits.
In Eq. (\ref{eq:OPImpl}), we use $\tilde{O}_{\mathcal{P}_{\rm temp}}$ at the first arrow and $O_=$'s at the second arrow, and the last operation is done by the multiply controlled NOT gate on the last $\kappa+1$ qubits.
Note that `1' on the last qubit means $(k,i)\notin\mathcal{P}_{\rm temp}$.
Therefore, the above operation is in fact $O_{\mathcal{P}_{\rm temp}}$, with some registers in Eq. (\ref{eq:OPImpl}) regarded as ancillas.
In this implementation, the QRAM $\tilde{O}_{\mathcal{P}_{\rm temp}}$ is queried once in a call to the oracle $O_{\mathcal{P}_{\rm temp}}$, along with $O(\kappa)$ uses of arithmetic oracles.
For initializing $\tilde{O}_{\mathcal{P}_{\rm temp}}$, $O(\kappa n)$ time is taken at the very beginning of Algorithm \ref{alg:QMCIbased}, where $\mathcal{P}_{\rm temp}=\emptyset$ and thus $k_{i,1}=\cdots=k_{i,\kappa}=-1$ for any $i\in[n]_0$.
After that, every time an index pair $(k,i)$ is added to $\mathcal{P}_{\rm temp}$ in the QAA loop in Algorithm \ref{alg:QMCIbased}, one memory cell in $\tilde{O}_{\mathcal{P}_{\rm temp}}$ is updated, which takes $O(1)$ time.

Let us summarize the above discussion.
At the beginning both of the naive iteration method and the QMCI method, we need to initialize the QRAMs $O_{\rm seq}$, $O_{\rm PWM}$ and $\tilde{O}_{\mathcal{P}_{\rm temp}}$, which takes $O(n+m|\mathcal{A}|K+\kappa n)$ time in total.
If we reasonably assume that $m|\mathcal{A}|K< n$ and $\kappa=O(1)$, the time complexity is estimated as $O(n)$.

Although we need to take $O(n)$ time at the preliminary stage, after that the quantum algorithms run with complexities shown in Theorems \ref{th:NaiveIter} and \ref{th:QMCIBased}, which scales with $n$ as $O(\sqrt{n})$, for any sequence $S$ and any PWMs $\{M_k\}_k$.
Also note that, once we prepare $O_{\rm seq}$, whose preparation is the bottleneck under the current assumption, we can search the matches between $S$ and another set of $K$ PWMs $\{M^\prime_k\}_k$ by preparing $O_{\rm PWM}$ and $\tilde{O}_{\mathcal{P}_{\rm temp}}$ and running the quantum algorithm, which no longer takes $O(n)$ time\footnote{Initializing $\tilde{O}_{\mathcal{P}_{\rm temp}}$ seems to take $O(\kappa n)$ time again. However, if we have $\tilde{O}_{\mathcal{P}_{\rm temp}}$ used in the previous algorithm run, resetting its updated memory cells gives us the properly initialized $\tilde{O}_{\mathcal{P}_{\rm temp}}$.
Since the number of the memory cells to be reset is equal to that of the matches found in the previous run and it is usually much smaller than $n$, this reset-based initialization does not take $O(n)$ time.}.
As far as the authors know, there is no known method for PWM matching in which initialization takes $O(n)$ time and the main search algorithm takes the sublinear complexity to $n$.
As an algorithm having the initialization cost of same order, we refer to \cite{beckstette2006fast}, for example.
In this classical algorithm based on an enhanced suffix array (ESA), it takes $O(n)$ time to construct ESA.
After that, the worst-case complexity to find matches is $O(n+m)$ if some condition is satisfied, but it can be $O(nm)$ in the general case.

\subsection{Score threshold in the large $m$ limit \label{sec:ScoreTol}}

Here, we consider the asymptotic distribution of scores of segment when the sequence motif length $m$ is large and, based on it, discuss the plausible setting on the score threshold.

In many cases, the score threshold $w_{\rm th}$ in matching with a PWM $M\in\mathbb{R}^{m\times|\mathcal{A}|}$ is determined by the $p$-value.
That is, we set $w_{\rm th}$ so that the probability that the score of a segment becomes equal to or larger than $w_{\rm th}$ is equal to the given value $p\in(0,1)$ in the background model.
Here, the background model means the assumption that, when we take a segment of length $m$ in the sequence $S$ randomly and denote by $u_i$ the alphabet in the $i$th position in the segment, $u_0,...,u_{m-1}$ are independent and identically distributed.
The rigorous definition is as follows.
Supposing that every $a\in\mathcal{A}$ is associated with $p_a\in(0,1)$ satisfying $\sum_{a\in\mathcal{A}}p_a=1$, we consider the finite probability space $(\mathcal{A}^m,\mathbb{P}_{\rm BG})$ consisting of the sample space $\mathcal{A}^m$ and the probability function $\mathbb{P}_{\rm BG}:\mathcal{A}^m\rightarrow\mathbb{R}_{\ge 0}$ such that, for any $u_0..u_{m-1}\in\mathcal{A}^m$,
\begin{equation}
    \mathbb{P}_{\rm BG}(u_0..u_{m-1})=\prod_{i=0}^{m-1}p_{a_i},
\end{equation}
if $u_0=a_0,...,u_{m-1}=a_{m-1}$ with $a_0,...,a_{m-1}\in\mathcal{A}$.
Then, we define
\begin{equation}
    w_{\rm th}:= \max\{w\in\mathbb{R} \ | \ \mathbb{P}_{\rm BG}\left(\{u_0..u_{m-1} \ | \ W_M(u_0..u_{m-1})\ge w\}\right)\ge p\},  
\end{equation}
where, for any subset $U\in\mathcal{A}^m$, we define $\mathbb{P}_{\rm BG}(U):=\sum_{u\in U}\mathbb{P}_{\rm BG}(u)$.

Now, we regard $W:=W_M(u_0..u_{m-1})$ as a random variable and consider its asymptotic distribution in the case of large $m$.
We use the following theorem, a variant of the central limit theorem.
\begin{theorem}[Theorem 27.4 in \cite{billingsley2008probability}, modified]
Let $\{X_n\}_{n\in\mathbb{N}_{\ge 0}}$ be the sequence of the independent random variables on some probability space $(\Omega,\mathcal{F},\mathbb{P})$ such that, for any $n\in\mathbb{N}_{\ge 0}$, $X_n$ has the expectation $\mu_n$ and the finite variance $\sigma_n^2$.
For each $n\in\mathbb{N}$, define $S_n:=\sum_{i=0}^{n-1} (X_i-\mu_i)$ and $s_n^2:=\sum_{i=0}^{n-1} \sigma_i^2$.
Suppose that there exists $\delta\in\mathbb{R}_+$ such that
\begin{equation}
    \lim_{n\rightarrow\infty} \frac{1}{s_n^{2+\delta}}\sum_{i=0}^{n-1} E_\mathbb{P}[|X_i-\mu_i|^{2+\delta}]=0. \label{eq:condCLT}
\end{equation}
Then, $S_n/s_n$ converges in distribution to a standard normal random variable, as $n$ goes to infinity.
\label{th:CLT}
\end{theorem}
\noindent We can apply this theorem to the case of PWM matching by, for each $i\in[m]_0$, regarding $X_i$ in Theorem \ref{th:CLT} as $M(i,u_i)$, the score of the alphabet in the $i$th position in the background model.
$\mu_i$ and $\sigma_i$ correspond to the mean and the variance of the score of the alphabet in the $i$th position, that is,
\begin{equation}
    \mu_i = \sum_{a\in\mathcal{A}} p_a M(i,a)
\end{equation}
and
\begin{equation}
    \sigma_i^2 = \sum_{a\in\mathcal{A}} p_a (M(i,a)-\mu_i)^2,
\end{equation}
respectively.
We must check the condition (\ref{eq:condCLT}) is satisfied, and it is in fact satisfied in the following plausible situation. 
First, recall that we have rescaled the PWM so that Eq. (\ref{eq:M01}) holds.
Therefore,
\begin{equation}
    E_{\mathbb{P}_{\rm BG}}[|X_i-\mu_i|^{2+\delta}]\le 1 \label{eq:condCLTforPWM1}
\end{equation}
holds obviously.
Besides, we may additionally assume that there exist $r\in(0,1)$ and $\sigma_{\rm min}^2\in\mathbb{R}_+$ independent of $m$ such that, for at least $\lceil rm \rceil$ elements $i$ in $[m]_0$, $\sigma_i^2\ge \sigma_{\rm min}^2$ holds.
This means that, although in some part of the positions the position-wise score variances might be small, at least in the certain ratio $r$ of the positions the variances exceeds the level $\sigma_{\rm min}^2$. 
This assumption yields
\begin{equation}
    s_m^2\ge rm\sigma_{\rm min}^2. \label{eq:condCLTforPWM2}
\end{equation}
Combining Eqs. (\ref{eq:condCLTforPWM1}) and (\ref{eq:condCLTforPWM2}), we have
\begin{equation}
    \frac{1}{s_m^{2+\delta}}\sum_{i=0}^{m-1} E_{\mathbb{P}_{\rm BG}}[|X_i-\mu_i|^{2+\delta}] \le \frac{m}{(rm)^{1+\delta/2}\sigma_{\rm min}^{2+\delta}}
\end{equation}
for any $\delta\in\mathbb{R}_+$, which converges to 0 in the large $m$ limit.

Therefore, in the large $m$ limit, we can approximate as
\begin{equation}
    \mathbb{P}_{\rm BG}\left(W\ge w_{\rm th} \right) \approx \int^\infty_{(w_{\rm th}-\tilde{\mu}_m)/s_m} \frac{1}{\sqrt{2\pi}} e^{-x^2/2} dx, \label{eq:scoreDistAsym}
\end{equation}
with $\tilde{\mu}_m:=\sum_{i=0}^{m-1}\mu_i=E_{\mathbb{P}_{\rm BG}}[W]$ being the expected segment score under the background model.

Considering the above asymptotic distribution of $S$, it seems reasonable to set $w_{\rm th}$ as $\tilde{\mu}_m+x\sigma_{\rm tot}$ with some $x\in\mathbb{R}_+$.
Alternatively, if we set the two levels of the score $w_{\rm hard}$ and $w_{\rm soft}$ in the QMCI-based method, it seems plausible to set them as $w_{\rm soft}=\tilde{\mu}_m+x_{\rm soft}s_m$ and $w_{\rm hard}=\tilde{\mu}_m+x_{\rm hard}s_m$ with $x_{\rm soft},x_{\rm hard}\in\mathbb{R}_+$ such that $x_{\rm soft}<x_{\rm hard}$ and $x_{\rm hard}-x_{\rm soft}=O(1)$.
For example, $x_{\rm soft}=3$ and $x_{\rm hard}=4$, which correspond to the $p$-values $1.35\times 10^{-3}$ and $3.17\times 10^{-5}$, respectively, in the approximation as Eq. (\ref{eq:scoreDistAsym}).
In such a setting, using $s_m=\Omega(\sqrt{m})$ that follows from Eq. (\ref{eq:condCLTforPWM2}), we get Eq. (\ref{eq:whardwsoftdiff}) and then the complexity bounds (\ref{eq:compQMCIOseq3}) and (\ref{eq:compQMCIOseq4}).

\section{Summary \label{sec:sum}}

In this paper, we have proposed the two quantum algorithms for an important but time-consuming problem in bioinformatics, PWM matching, which aims to find sequence motifs in a biological sequence whose scores defined by PWMs exceed the threshold.
Both of these algorithms, the naive iteration method and the QMCI-based method, utilize QAA for search of high-score segments.
They are differentiated by how to calculate the segment score.
The former calculates it by simply iterating to add up each position-wise score by quantum circuits for arithmetic.
The latter uses QMCI for this summation, coping with false detection due to the QMCI error by setting two levels of threshold $w_{\rm soft}$ and $w_{\rm hard}$.
Given the oracular accesses to the entries in the sequence and PWMs, both of the quantum algorithms runs with query complexity scaling with the sequence length $n$ and the number of PWMs $K$ as $O(\sqrt{Kn})$, thanks to the well-known quadratic speedup by QAA.
Furthermore, under some setting on $w_{\rm soft}$ and $w_{\rm hard}$, the complexity of the QMCI-based method scales with the sequence motif length $m$ as $O(\sqrt{m})$.
These mean the quantum speedup over existing classical algorithms.
Although our quantum algorithms take $O(n)$ preparation time for initialization of QRAMs, they still have the advantage especially when we perform matching between a sequence and many PWMs.

It is interesting that these algorithms for bioinformatics are inspired by the algorithm in \cite{Miyamoto2022} for a completely different field, gravitational wave astronomy.
We expect that the scheme used in these algorithms, the combination of QMCI and QAA, is widely useful over various industrial and scientific fields.
In future work, we will explore other applications of quantum algorithms in bioinformatics and other fields.

\section*{Acknowledgement}

This work is supported by MEXT Quantum Leap Flagship Program (MEXT Q-LEAP) Grant no. JPMXS0120319794 and JPMXS0118067285.
K.M. is supported by Japan Society for the Promotion of Science (JSPS) KAKENHI Grant no. JP22K11924.

\appendix

\section{Proof of Theorem \ref{th:QAA} for the case that $a=0$ \label{sec:casea0}}

\begin{proof}
    As described in \cite{Miyamoto2022} and the original paper \cite{brassard2002}, in QAA, we repeatedly generate $G^j\ket{\Phi}$ with various $j\in\mathbb{N}_{\ge 0}$ and measure $R_2$, and output (A) if and only if the measurement outcome is 1.
    Here,
    \begin{equation}
        G:=-AS_0A^{-1}S_\chi,
    \end{equation}
    where $S_0$ and $S_\chi$ are the unitary operators on the system under consideration acting as follows:
    \begin{equation}
        S_\chi\ket{\phi}\ket{x}=
        \begin{cases}
            \ket{\phi}\ket{0} & ; \ {\rm if} \ x=0 \\ 
            -\ket{\phi}\ket{1} & ; \ {\rm if} \ x=1
        \end{cases}
    \end{equation}
    with any state $\ket{\phi}$ on $R_1$, and
    \begin{equation}
    S_0\ket{\Phi^\prime}=
    \begin{cases}
        -\ket{0}\ket{0} & ; \ {\rm if} \ \ket{\Phi^\prime}=\ket{0}\ket{0} \\ 
        \ket{\Phi^\prime} & ; \ {\rm if} \ \bra{\Phi^\prime}\left(\ket{0}\ket{0}\right) = 0
    \end{cases}.
    \end{equation}
    As we can see easily, if $a=0$, $G\ket{\Phi}=\ket{\Phi}=\ket{\phi_0}\ket{0}$ holds, and thus we never get 1 in measuring $R_2$ in $G^j\ket{\Phi}$ for any $j$.
    Therefore, (A) is never output, which means that (B) is output certainly.
\end{proof}

For the detailed procedure of $\proc{QAA}(A,\gamma,\delta)$ and the proof of Theorem \ref{th:QAA} in the other cases, see \cite{Miyamoto2022} and the original paper \cite{brassard2002}.

\section{Proof of Theorem \ref{th:QMCI} \label{sec:PrThQMCI}}

Before the proof, we present the following fact, Theorem 4 in \cite{Miyamoto2022}.
It is almost same as Lemma 2.1 in \cite{Montanaro2015}, but slightly modified in reference to the original one, Lemma 6.1 in \cite{JERRUM1986169}.

\begin{lemma}
    Let $\mu\in\mathbb{R}$ and $\epsilon\in\mathbb{R}_+$.
    Let $\mathcal{A}$ be an algorithm that outputs an $\epsilon$-approximation of $\mu$ with probability $\gamma\ge \frac{3}{4}$.
    Then, for any $\delta\in(0,1)$, the median of outputs in $12\left\lceil\log\delta^{-1}\right\rceil+1$ runs of $\mathcal{A}$ is an $\epsilon$-approximation of $\mu$ with probability at least $1-\delta$.
    \label{lem:median}
\end{lemma}


Then, the proof of Theorem \ref{th:QMCI} is as follows.
\begin{proof}[Proof of Theorem \ref{th:QMCI}]
    According to Theorem 7 in \cite{Miyamoto2022}, for any integer $t$ larger than 2, we can construct the oracle $\tilde{O}_{\mathcal{X},t}^{\rm mean}$ that acts as $\tilde{O}_{\mathcal{X},t}^{\rm mean}\ket{0}=\sum_{y\in \mathcal{Y}} \alpha_y\ket{y}$ using $O_{\mathcal{X}}$ $O(t)$ times.
    Here, $\mathcal{Y}$ is a finite set of real numbers that includes a subset $\tilde{\mathcal{Y}}$ consisting of elements $\tilde{\mu}$ satisfying
    \begin{equation}
        |\tilde{\mu}-\mu|\le C\left(\frac{\sqrt{\mu}}{t}+\frac{1}{t^2}\right)
    \end{equation}
    with a universal real constant $C$, and $\{\alpha_y\}_{y\in\mathcal{Y}}$ are complex numbers satisfying $\sum_{\tilde{y}\in\tilde{\mathcal{Y}}}|\alpha_{\tilde{y}}|^2\ge 8/\pi^2$.
    Following this, we prepare a system with $J$ quantum registers and generate the state
    \begin{equation}
        \ket{\Psi}:=\left(\sum_{y_1\in \mathcal{Y}} \alpha_{y_1}\ket{y_1}\right)\otimes\cdots\otimes\left(\sum_{y_J\in \mathcal{Y}} \alpha_{y_J}\ket{y_J}\right)
    \end{equation}
    by operating $\tilde{O}_{\mathcal{X},t}^{\rm mean}$ on each register.
    Here, $J$ and $t$ are set as
    \begin{equation}
    J=12\left\lceil\log\delta^{-1}\right\rceil+1,t=\left\lceil\frac{2C}{\epsilon}\right\rceil.
    \end{equation}
    It can be shown by easy algebra that, in this setting, each $\tilde{\mu}\in\tilde{\mathcal{Y}}$ satisfies $|\tilde{\mu}-\mu|\le\epsilon$.
    By measuring $\ket{\Psi}$, we obtain $J$ real numbers $y_1,...,y_J$, each of which is a $\epsilon$-approximation of $\mu$ with probability at least $\frac{8}{\pi^2}>\frac{3}{4}$.
    Therefore, because of Lemma \ref{lem:median}, the median of $y_1,...,y_J$ is a $\epsilon$-approximation of $\mu$ with probability at least $1-\delta$.
    This means that, if we generate
    \begin{equation}
        \ket{\Psi^\prime}:=\left(\sum_{y_1\in \mathcal{Y}} \alpha_{y_1}\ket{y_1}\right)\otimes\cdots\otimes\left(\sum_{y_J\in \mathcal{Y}} \alpha_{y_J}\ket{y_J}\right)\ket{{\rm med}(y_1,...,y_J)}
    \end{equation}
    by adding one more register to $\ket{\Psi}$ and then using $O^{\rm med}_J$, this is actually the state in Eq. (\ref{eq:OmeanX}), with the first $J$ registers regarded as undisplayed.
    In summary, we can construct $O_{\mathcal{X},\epsilon,\delta}^{\rm mean}$ as
    \begin{equation}
        O_{\mathcal{X},\epsilon,\delta}^{\rm mean}=O^{\rm med}_J\left(\underbrace{\tilde{O}_{\mathcal{X},t}^{\rm mean}\otimes\cdots\otimes \tilde{O}_{\mathcal{X},t}^{\rm mean}}_{J}\otimes I\right),
    \end{equation}
    where $I$ is the identity operator on the Hilbert space for the register.
    Since each $\tilde{O}_{\mathcal{X},t}^{\rm mean}$ uses $O_{\mathcal{X}}$ $O(t)$ times, $O_{\mathcal{X},\epsilon,\delta}^{\rm mean}$ uses $O_{\mathcal{X}}$ $O(tJ)$ times, that is, $O\left(\frac{1}{\epsilon}\log\left(\frac{1}{\delta}\right)\right)$ times, in total.
\end{proof}

\bibliographystyle{apsrev4-2}
\bibliography{reference}

\end{document}